\documentclass[11pt]{amsart}
\usepackage{amsmath,amssymb,amsthm}
\usepackage[margin=1.1in]{geometry}
\usepackage{xcolor}
\usepackage{tikz}
\usetikzlibrary{arrows.meta,positioning,decorations.pathreplacing,calc}
\usepackage{pgfplots}
\definecolor{gmblue}{HTML}{0072B2}
\definecolor{gmverm}{HTML}{D55E00}
\definecolor{gmgreen}{HTML}{009E73}
\definecolor{gmpurp}{HTML}{AA3377}
\definecolor{gmgold}{HTML}{B8860B}
\pgfplotsset{compat=1.17,
  gmaxis/.style={width=0.47\textwidth, height=5.2cm,
    tick label style={font=\scriptsize}, label style={font=\scriptsize},
    title style={font=\scriptsize},
    legend style={font=\tiny, draw=none, fill=none}, legend cell align=left,
    every axis plot/.append style={line width=0.6pt}}}
\usepackage{hyperref}
\hypersetup{colorlinks=true,linkcolor=black,citecolor=black,urlcolor=black}

\newtheorem{theorem}{Theorem}[section]
\newtheorem{lemma}[theorem]{Lemma}
\newtheorem{proposition}[theorem]{Proposition}
\newtheorem{corollary}[theorem]{Corollary}
\theoremstyle{definition}
\newtheorem{definition}[theorem]{Definition}
\newtheorem{remark}[theorem]{Remark}

\newcommand{\Var}{\operatorname{Var}}
\newcommand{\E}{\mathbb{E}}
\newcommand{\R}{\mathbb{R}}
\newcommand{\C}{\mathbb{C}}
\newcommand{\Z}{\mathbb{Z}}
\newcommand{\Unif}{\operatorname{Unif}}
\newcommand{\diag}{\operatorname{diag}}
\newcommand{\re}{\operatorname{Re}}
\newcommand{\one}{\mathbf{1}}
\newcommand{\Bool}{\mathbb{B}}

\title[Depth analysis of GM-QAOA]{Depth analysis of the Quantum Approximate Optimization Algorithm with a Grover mixer}

\author{Bojko N. Bakalov}
\address{Department of Mathematics, North Carolina State University,
Raleigh, NC 27695, USA}
\email{bnbakalo@ncsu.edu}

\author{Dimitar Grantcharov}
\address{Department of Mathematics, University of Texas at Arlington, Arlington, TX 76021, USA}
\email{grandim@uta.edu}

\date{September 15, 2026}

\begin{document}

\begin{abstract}
We study the Quantum Approximate Optimization Algorithm with a Grover mixer and independently sampled cost and mixing angles. Under a lattice condition on the cost values carried by the initial state, we prove a depth-independent lower bound for the variance of the loss at every depth, together with the same bound for the derivative with respect to the final mixing angle. The estimate is instance-dependent and, for fixed locality, is inverse polynomial in the number of qubits for integer-valued local objective functions with uniformly bounded local terms, in particular for MaxCut. We also establish Grover-type reachability bounds showing that the depth required to approximate a prescribed carried eigenspace is bounded below by a constant multiple of the inverse square root of its initial probability.
\end{abstract}

\maketitle

\section{Introduction}\label{sec:intro}

Variational quantum algorithms use parametrized quantum circuits together with a classical optimization loop to minimize the expectation value of a problem-dependent observable \cite{Cerezo21}. The Quantum Approximate Optimization Algorithm (QAOA), introduced by Farhi, Goldstone, and Gutmann \cite{FGG14}, is a central example for combinatorial optimization. It alternates evolutions generated by a cost Hamiltonian and a mixer Hamiltonian, with the evolution parameters optimized classically. A basic obstacle to this procedure is the barren-plateau phenomenon, in which the gradients of the loss function become exponentially concentrated near zero as the size of the problem increases \cite{MBSBN18,Larocca25}. Quantitative control of loss and gradient fluctuations is therefore a fundamental question in the analysis of QAOA trainability.

We consider the Grover-mixer variant of QAOA (GM-QAOA), introduced by B\"artschi and Eidenbenz \cite{BE20}, in which the mixer Hamiltonian is the projector $G=|\xi\rangle\langle\xi|$ onto the initial state $|\xi\rangle$ and the cost Hamiltonian $H$ is diagonal in the computational basis. This mixer gives a particularly rigid dynamics and is natural for constrained problems in which $|\xi\rangle$ is chosen inside the feasible subspace \cite{BE20,HWORVB19}. It also makes the dependence of trainability on the circuit depth especially transparent. In \cite{TNB25}, the authors determined the dynamical Lie algebra \cite{LCSMCC22,Ragone24} of GM-QAOA. They showed that, for integer-valued $s$-local objective functions, the variance of the normalized loss is bounded below by an inverse polynomial for \emph{sufficiently large depth} (see Theorem~IV.4 of \cite{TNB25}). The variance of the loss is closely related to the usual gradient formulation of barren plateaus \cite{Larocca25,Ragone24}, and under the assumptions of \cite{AHCC22}, exponential concentration of the loss implies exponential concentration of its gradients.

However, the result of \cite{TNB25} is asymptotic in depth and does not give an explicit threshold at which the deep-circuit regime begins. This leaves open whether exponentially small loss or gradient fluctuations can occur at finite depth. Numerical evidence that Lie-algebraic variance predictions do not describe shallow circuits, for QAOA applied to the maximum independent set problem, is given by Copp et al.~\cite{CLMQ26}. We address this question directly, without using dynamical Lie algebras, and obtain a lower bound that is valid at \emph{every depth}. In order to state the result, we first introduce the necessary notation.

For $\beta,\gamma\in\R$, put
\begin{equation}\label{eq:ZR}
Z(\gamma):=e^{-i\gamma H},\qquad R(\beta):=e^{-i\beta G}=I+(e^{-i\beta}-1)\,G.
\end{equation}
Starting with the initial state $|\psi_0\rangle:=|\xi\rangle$, a GM-QAOA circuit (also called an ansatz) of depth $p\ge1$ produces the state
(see Figure~\ref{fig:gmqaoa}):
\begin{equation}\label{eq:psip}
|\psi_p\rangle:=R(\beta_p)Z(\gamma_p)\cdots R(\beta_1)Z(\gamma_1)|\xi\rangle.
\end{equation}
\begin{figure}[t]
\centering
\begin{tikzpicture}[
  >=Latex,
  gate/.style={draw=black, rounded corners=2pt, line width=0.75pt,
    minimum width=1.55cm, minimum height=1.00cm, align=center, fill=black!4},
  state/.style={inner sep=1pt},
  note/.style={font=\scriptsize, align=center},
  every path/.style={draw=black}
]
\coordinate (a0) at (0,0);
\coordinate (a1) at (1.15,0);
\coordinate (a2) at (3.15,0);
\coordinate (a3) at (5.15,0);
\coordinate (a4) at (6.35,0);
\coordinate (a5) at (8.35,0);
\coordinate (a6) at (10.35,0);
\coordinate (a7) at (11.65,0);

\node[state, left=0.12cm of a0] (in) {$|\xi\rangle$};
\node[gate] (z1) at (a1) {$Z(\gamma_1)$\\[-1pt]{\scriptsize $e^{-i\gamma_1H}$}};
\node[gate] (r1) at (a2) {$R(\beta_1)$\\[-1pt]{\scriptsize $e^{-i\beta_1G}$}};
\node[state] (dots) at (a4) {$\cdots$};
\node[gate] (zp) at (a5) {$Z(\gamma_p)$\\[-1pt]{\scriptsize $e^{-i\gamma_pH}$}};
\node[gate] (rp) at (a6) {$R(\beta_p)$\\[-1pt]{\scriptsize $e^{-i\beta_pG}$}};
\node[state, right=0.10cm of a7] (out) {$|\psi_p\rangle$};

\draw[-{Latex[length=2.1mm]}, line width=0.8pt] (in.east) -- (z1.west);
\draw[-{Latex[length=2.1mm]}, line width=0.8pt] (z1.east) -- (r1.west);
\draw[-{Latex[length=2.1mm]}, line width=0.8pt] (r1.east) -- (dots.west);
\draw[-{Latex[length=2.1mm]}, line width=0.8pt] (dots.east) -- (zp.west);
\draw[-{Latex[length=2.1mm]}, line width=0.8pt] (zp.east) -- (rp.west);
\draw[-{Latex[length=2.1mm]}, line width=0.8pt] (rp.east) -- (out.west);

\draw[decorate,decoration={brace,amplitude=4pt},line width=0.7pt]
  ($(z1.south west)+(0,-0.22)$) -- ($(r1.south east)+(0,-0.22)$)
  node[midway,below=5pt,note] {layer $1$};
\draw[decorate,decoration={brace,amplitude=4pt},line width=0.7pt]
  ($(zp.south west)+(0,-0.22)$) -- ($(rp.south east)+(0,-0.22)$)
  node[midway,below=5pt,note] {layer $p$};

\node[note, above=0.17cm of z1] {cost evolution};
\node[note, above=0.17cm of r1] {Grover mixer};
\node[note, above=0.17cm of zp] {cost evolution};
\node[note, above=0.17cm of rp] {Grover mixer};
\end{tikzpicture}
\caption{Schematic of the depth-$p$ GM-QAOA circuit. Each layer consists of a cost evolution $Z(\gamma_k)$ followed by the Grover-mixer evolution $R(\beta_k)$.}
\label{fig:gmqaoa}
\end{figure}
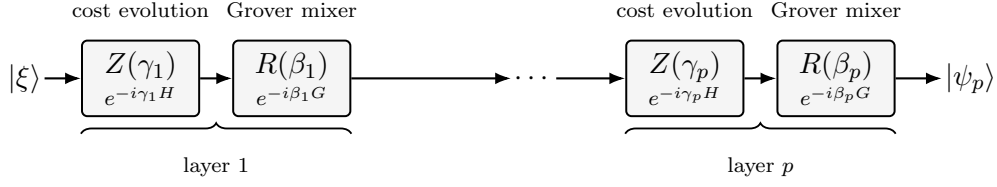
The loss function is defined as the expectation value of the observable $H$:
\begin{equation*}
\ell_p:=\langle\psi_p|H|\psi_p\rangle .
\end{equation*}
Write the spectral decomposition of the Hamiltonian as $H=\sum_\lambda\lambda\Pi_\lambda$.
Let $\lambda_1<\cdots<\lambda_d$ be the distinct eigenvalues $\lambda$ of $H$ for which $\Pi_{\lambda}|\xi\rangle\ne0$, and set
\begin{equation}\label{eq:cj}
c_j:=\sqrt{\langle\xi|\Pi_{\lambda_j}|\xi\rangle}>0, \qquad |\xi_j\rangle:=c_j^{-1}\Pi_{\lambda_j}|\xi\rangle, \qquad 1\le j\le d.
\end{equation}
With this notation, we have
\begin{equation}\label{eq:xi}
|\xi\rangle = \sum_{j=1}^d c_j |\xi_j\rangle, \qquad
\sum_{j=1}^d c_j^2 = 1, \qquad
H|\xi_j\rangle = \lambda_j|\xi_j\rangle,
\end{equation}
and the eigenvectors $|\xi_j\rangle$ are orthonormal (cf.\ \cite{TNB25}). The expectation and variance of $H$ in the initial state are
\begin{equation}\label{eq:b0}
\ell_0=\langle\xi|H|\xi\rangle=\sum_{j=1}^dc_j^2\lambda_j, \qquad
b_0^2:= \langle\xi|H^2|\xi\rangle-\langle\xi|H|\xi\rangle^2 = \sum_{j=1}^dc_j^2(\lambda_j-\ell_0)^2.
\end{equation}
Assume $d\ge2$ and that $\lambda_j-\lambda_k\in\eta\Z$ for all $j,k$, where $\eta>0$. 
For a GM-QAOA circuit in depth-$p$, we sample angles $\beta_k,\gamma_k$ uniformly and independently, with $\beta_k\sim\Unif[0,2\pi)$ and $\gamma_k\sim\Unif[0,2\pi/\eta)$.

The following is the first main result of the paper.

{
\begin{theorem}\label{thm:intro-variance}
With the above notation, we have, for every $p\ge1$,
\[
\Var[\ell_p] \ge\ \frac{49}{384\pi}\,\frac{\eta\,b_0^4}{d^2(\lambda_d-\lambda_1)^3} \,.
\]
\end{theorem}
}

Corollary~\ref{cor:gradient} below gives the same lower bound for the variance of the derivative of $\ell_p$ with respect to the final mixing angle $\beta_p$. The estimate is instance-dependent. An inverse-polynomial barren-plateau conclusion for a growing family requires the right-hand side of the estimate to be bounded below by an inverse polynomial in the problem size. Theorem~\ref{thm:local} verifies this condition for integer-valued $s$-local objective functions with bounded local terms and the uniform-superposition initial state, and thereby gives an inverse-polynomial lower bound at every depth. It is a finite-depth counterpart of Theorem~IV.4 of \cite{TNB25}, which gives an inverse-polynomial bound for a sufficiently large depth. Corollary~\ref{cor:maxcut} specializes Theorem~\ref{thm:local} to MaxCut.

The argument consists of three parts. The variance of the loss with respect to the last mixing angle alone is an explicit multiple of the product of two squared overlaps (Lemma~\ref{lem:lastmixer}). These are the overlaps of the state entering the last mixer with the initial state and with a fixed second vector. The previous mixer adds to the state a multiple of the initial state whose coefficient depends affinely on the uniformly distributed phase $e^{-i\beta}$. An averaging inequality (Lemma~\ref{lem:phase}) shows that this addition cannot be canceled, on average, by the rest of the state. Finally, the mean squared overlap, or mean return probability, of the state with the initial state is at least $1/d$ at every depth, by an explicit linear recursion for the mean populations of the cost eigenspaces (Lemma~\ref{lem:populations}). The depth-one variance is bounded below in Proposition~\ref{prop:onelayer}, and Theorem~\ref{thm:main} combines this estimate with Lemmas~\ref{lem:lastmixer}--\ref{lem:populations} to obtain the bound at arbitrary depth.

We also study the reachability of GM-QAOA. QAOA reachability deficits, that is, the failure of the ansatz to reach the optimum below a threshold depth, were identified by Akshay et al.~\cite{APMB20}. The authors of \cite{TNB25} observe, as a consequence of their Lie-algebraic result, that the normalized projection of the initial state onto a carried eigenspace can be approximated to arbitrary precision at sufficiently large depth. Fix an eigenvalue $\lambda_* = \lambda_{j_0}$ for some $1\le j_0\le d$, and denote by
\begin{equation*}
|e\rangle:=|\xi_{j_0}\rangle = \frac{\Pi_{\lambda_*}|\xi\rangle}{\|\Pi_{\lambda_*}|\xi\rangle\|}
\end{equation*}
the corresponding normalized eigenstate. We also let
\begin{equation*}
\sin\theta_0:=\|\Pi_{\lambda_*}|\xi\rangle\| = c_{j_0} \in(0,1),
\qquad
\theta_0\in(0,\pi/2).
\end{equation*}
In the reachability results, the angles $\beta_k,\gamma_k$ are arbitrary real numbers rather than random variables.

Our second main result gives a lower bound on the depth required to approximate the chosen eigenstate $|e\rangle$.

\begin{theorem}
\label{thm:intro-depth}
Let\/ $\varepsilon\in(0,1)$. If\/
$|\langle e|\psi_p\rangle|^2\ge1-\varepsilon$ for some choice of
parameters $\beta_k,\gamma_k$, then
\[
p\ge
\frac{\arcsin\sqrt{1-\varepsilon}}{2\theta_0}-\frac12 \,.
\]
In particular, if\/ $|\langle e|\psi_p\rangle|^2=1$, then
\[
p\ge
\frac{\pi}{4\theta_0}-\frac12
\ge
\frac12\left(\frac1{\sin\theta_0}-1\right).
\]
If, in addition, $|\xi\rangle=|+\rangle^{\otimes n}$, $H$ is
diagonal in the computational basis, and $\lambda_*$ is attained on
exactly $N_*$ basis states, then
\[
p\ge
\frac12\bigl(2^{n/2}N_*^{-1/2}-1\bigr).
\]
\end{theorem}

Theorem~\ref{thm:grover} below proves the Grover envelope underlying Theorem~\ref{thm:intro-depth} and bounds the probability of every set of carried cost values by its initial probability multiplied by a quadratic function of depth. For uniform feasible-state initialization, Theorem~\ref{thm:grover}(3) recovers, in the common GM-QAOA setting, the $(2p+1)^2$ amplification bound of Bridi and de Lima Marquezino \cite{BM24}. The pointwise and optimal-set special cases of this bound were obtained by Xie et al.~\cite{Xie25}. In particular, the last bound in Theorem~\ref{thm:intro-depth} has the same inverse-square-root scaling as these earlier bounds. The binary-spectrum case of Theorem~\ref{thm:grover}(1) is the variational Grover setting studied in \cite{BM24,Li24}. The threshold variant of GM-QAOA introduced by Golden et al.~\cite{GBOE21} can be viewed as a generalization of Grover search to approximate optimization, and Benchasattabuse et al.~\cite{BBGLE25} obtained depth lower bounds for guaranteed approximation ratios with the Grover mixer. In the present paper, the initial spectral weights and the carried spectrum are arbitrary. Theorem~\ref{thm:inject} gives a spectrum-dependent refinement, and Theorem~\ref{thm:run} gives an independent obstruction to exact reachability when the carried spectrum contains a run of consecutive lattice points. Related consequences of expressivity and overparameterization \cite{LJGCC23} for the complete-graph quantum-walk-based optimization algorithm \cite{MW20}, which is equivalent to GM-QAOA up to a rescaling of the mixing parameter, were obtained in \cite{BLPSMA26}. Theorems~\ref{thm:intro-variance} and~\ref{thm:intro-depth} together show that a depth-independent lower bound on the loss variance is compatible with an exponentially large depth requirement for concentrating on an optimum whose initial probability is exponentially small.

The paper is organized as follows. Section~\ref{sec:setting} fixes the setting and notation. Section~\ref{sec:proof-main} proves Theorem~\ref{thm:intro-variance}. Section~\ref{sec:applications} gives applications of Theorem~\ref{thm:intro-variance}. Section~\ref{sec:reach} proves Theorem~\ref{thm:intro-depth} and develops further reachability bounds. Finally, Section~\ref{sec:conclusions} gives the conclusions.

\section{Setting and notation}\label{sec:setting}

{In this section, we fix the setting and notation. We follow the notation of \cite{TNB25} whenever it applies.}

\subsection{The GM-QAOA ansatz}

Let $V$ be a complex finite-dimensional Hilbert space, $H$ be a Hermitian operator on $V$, and $|\xi\rangle\in V$ be a fixed unit vector. Write $H=\sum_{\lambda}\lambda\,\Pi_\lambda$ for the spectral decomposition of $H$, and let $\lambda_1<\lambda_2<\dots<\lambda_d$ be the eigenvalues $\lambda$ with $\Pi_\lambda|\xi\rangle\ne0$. For $j=1,\dots,d$, we define $c_j$ and $|\xi_j\rangle$ by \eqref{eq:cj}; then \eqref{eq:xi} holds.
We denote by
\begin{equation*}
W_0:=\operatorname{span}\bigl\{|\xi_1\rangle,\dots,|\xi_d\rangle\bigr\}
\end{equation*}
the subspace of $V$ with orthonormal basis $\{|\xi_j\rangle\}$.
Throughout, we assume $d\ge2$.

The Grover mixer is the orthogonal projector $G:=|\xi\rangle\langle\xi|$. 
For a depth $p\ge1$ and real parameters $\beta=(\beta_1,\dots,\beta_p)$, $\gamma=(\gamma_1,\dots,\gamma_p)$, the GM-QAOA ansatz prepares
the state $|\psi_p\rangle$ given by \eqref{eq:psip} (see also \eqref{eq:ZR}),
and evaluates the loss $\ell_p:=\langle\psi_p|H|\psi_p\rangle$.
The alternating circuit is shown schematically in Figure~\ref{fig:gmqaoa}.

We also write $|\psi_0\rangle:=|\xi\rangle$ and
\begin{equation}\label{eq:phikok}
|\varphi_k\rangle:=Z(\gamma_k)|\psi_{k-1}\rangle,\qquad o_k:=\langle\xi|\varphi_k\rangle, \qquad 1\le k\le p,
\end{equation}
so that $|\varphi_k\rangle$ is the state entering the $k$-th mixer and $|\psi_k\rangle=R(\beta_k)|\varphi_k\rangle$.

\subsection{Parameter distribution and standing assumptions}

\begin{definition}\label{def:lattice}
The pair $(H,|\xi\rangle)$ satisfies the \emph{lattice condition with spacing $\eta>0$} if $\lambda_j-\lambda_k\in\eta\Z$ for all $1\le j,k\le d$.
\end{definition}

The lattice condition holds with $\eta=1$ whenever $H$ is a diagonal operator with integer entries, which is the case for integer-valued objective functions. Under the lattice condition, the operators $Z(\gamma+2\pi/\eta)$ and $Z(\gamma)$ agree on $W_0$ up to a global phase, so a cost angle $\gamma_k$ is naturally an element of $[0,2\pi/\eta)$.
Similarly, the mixer angles $\beta_k\in [0,2\pi)$.

\begin{definition}\label{def:law}
Under the lattice condition with spacing $\eta$, we sample the parameters $\beta_1,\dots,\beta_p$, $\gamma_1,\dots,\gamma_p$ independently and uniformly with $\beta_k\sim\Unif[0,2\pi)$ and $\gamma_k\sim\Unif[0,2\pi/\eta)$. We write $\Var_p:=\Var[\ell_p]$ for the variance of the loss $\ell_p$ under this law.
\end{definition}

\begin{definition}\label{def:data}
Set
\[
|e_1\rangle:=b_0^{-1}(H-\ell_0)|\xi\rangle,\qquad L:=\lambda_d-\lambda_1,
\]
where $\ell_0$ and $b_0$ are given by \eqref{eq:b0}.
In addition, put $c^2:=(c_1^2,\dots,c_d^2)^\top\in\R^d$ and
\[
W:=\diag(c^2)-c^2(c^2)^\top,\qquad M_1:=I-2W .
\]
\end{definition}

Since $d\ge2$ and all $c_j>0$, we have $b_0>0$ and $L>0$, so $|e_1\rangle$ is well defined. Moreover, we have $\langle\xi|e_1\rangle=b_0^{-1}(\langle\xi|H|\xi\rangle-\ell_0)=0$ and $\|e_1\|^2=b_0^{-2}\sum_jc_j^2(\lambda_j-\ell_0)^2=1$, so $|e_1\rangle$ is a unit vector orthogonal to $|\xi\rangle$. For $x\in\R^d$, we have 
\begin{equation}\label{eq:xWx}
x^\top Wx=\sum_{j=1}^d c_j^2x_j^2-\Bigl(\sum_{j=1}^d c_j^2x_j\Bigr)^2, 
\end{equation}
which is the variance of the coordinates of $x$ under the weights $c_j^2$. Hence, $W$ is positive semidefinite and $W\one=0$ for the all-ones vector $\one$. The matrix $M_1$ is symmetric and satisfies $M_1\one=\one$. Lemma~\ref{lem:populations}(3) below shows that $M_1$ is also positive semidefinite.

\begin{definition}\label{def:charfun}
For $\gamma\in\R$, put
\begin{align*}
s(\gamma)&:=\langle\xi|Z(\gamma)|\xi\rangle=\sum_{j=1}^d c_j^2e^{-i\gamma\lambda_j},\\ 
t(\gamma)&:=\langle(H-\ell_0)\xi|Z(\gamma)|\xi\rangle=\sum_{j=1}^d c_j^2(\lambda_j-\ell_0)e^{-i\gamma\lambda_j},
\end{align*}
and $M:=2\,\E_{\gamma\sim\Unif[0,2\pi/\eta)}\bigl[|s(\gamma)|^2|t(\gamma)|^2\bigr]$.
\end{definition}

The function $\gamma\mapsto s(-\gamma)$ is the characteristic function of the cost distribution $(c_j^2)_j$ on the values $(\lambda_j)_j$, and we have $t(\gamma)=i\,s'(\gamma)-\ell_0s(\gamma)$.
Characteristic-function methods for Grover-driven QAOA were used by Headley and Wilhelm \cite{HW23}.

\begin{lemma}[\cite{TNB25}]\label{lem:invariant}
The subspace $W_0$ is invariant under $H$, under $Z(\gamma)$ for every $\gamma$, and under $R(\beta)$ for every $\beta$. On $W_0$, in the orthonormal basis $\{|\xi_j\rangle\}$, we have $H=\diag(\lambda_1,\dots,\lambda_d)$ and $|\xi\rangle=(c_1,\dots,c_d)^\top$. In particular, we have $|\psi_p\rangle\in W_0$ for all $p$ and all parameters.
\end{lemma}

\begin{proof}
Each $|\xi_j\rangle$ is an eigenvector of $H$ with eigenvalue $\lambda_j$, so $H$ and $Z(\gamma)=e^{-i\gamma H}$ preserve $W_0$ and act on the basis $\{|\xi_j\rangle\}$ as stated. Since $|\xi\rangle=\sum_jc_j|\xi_j\rangle\in W_0$, the operator $R(\beta)=I+(e^{-i\beta}-1)\,G$ also preserves $W_0$. The last claim follows by induction on $p$.
\end{proof}

By Lemma~\ref{lem:invariant}, we can identify $W_0$ with $\C^d$, $H$ with $\diag(\lambda)$, $|\xi\rangle$ with the vector $(c_j)_j$, and the loss with $\ell_p=\sum_j\lambda_j|\langle\xi_j|\psi_p\rangle|^2$. We make these identifications throughout.
For a detailed discussion of this invariant-subspace reduction and of the restricted Hamiltonians on $W_0$, see Section~III.A of \cite{TNB25}.

\section{\texorpdfstring{Proof of Theorem~\ref{thm:intro-variance}}{Proof of Theorem 1.1}}\label{sec:proof-main}

We first establish three preliminary lemmas in Subsection~\ref{sec:lemmas} and a lower bound for the depth-one variance in Subsection~\ref{sec:onelayer}. These results are combined in Subsection~\ref{sec:main} to prove Theorem~\ref{thm:intro-variance} for arbitrary depth.

\subsection{Three lemmas}\label{sec:lemmas}

In this subsection, we prove three lemmas: an exact formula for the variance over one mixing angle, a phase-averaging inequality, and a recursion for the mean populations of the cost eigenspaces.

\begin{lemma}\label{lem:lastmixer}
Let $|\varphi\rangle\in W_0$ be a unit vector and $\beta\sim\Unif[0,2\pi)$. Then
\[
\langle R(\beta)\varphi|H|R(\beta)\varphi\rangle=\langle\varphi|H|\varphi\rangle+2\re\bigl[(e^{-i\beta}-1)\,z\bigr],\qquad z:=b_0\,\langle\xi|\varphi\rangle\langle\varphi|e_1\rangle,
\]
and consequently
\[
\Var_\beta\bigl[\langle R(\beta)\varphi|H|R(\beta)\varphi\rangle\bigr]=2|z|^2=2\,b_0^2\,|\langle\xi|\varphi\rangle|^2\,|\langle e_1|\varphi\rangle|^2 .
\]
\end{lemma}

\begin{proof}
Write $w:=e^{-i\beta}-1$ and $o:=\langle\xi|\varphi\rangle$, so that $R(\beta)|\varphi\rangle=|\varphi\rangle+wo|\xi\rangle$. Expanding the quadratic form, we obtain
\[
\langle R(\beta)\varphi|H|R(\beta)\varphi\rangle=\langle\varphi|H|\varphi\rangle+2\re\bigl[w\,o\,\langle\varphi|H|\xi\rangle\bigr]+|w|^2|o|^2\langle\xi|H|\xi\rangle .
\]
Since $|w|^2=2-2\cos\beta=-2\re w$ and $\langle\xi|H|\xi\rangle=\ell_0$, the last term is equal to $2\re\bigl[w\,(-|o|^2\ell_0)\bigr]=2\re\bigl[w\,o\,(-\bar o\,\ell_0)\bigr]$. Combining the two $\re$ terms gives $2\re\bigl[w\,o\,\langle\varphi|(H-\ell_0)|\xi\rangle\bigr]$, and $\langle\varphi|(H-\ell_0)|\xi\rangle=b_0\langle\varphi|e_1\rangle$ by the definition of $|e_1\rangle$. This proves the first identity with $z=b_0\,o\,\langle\varphi|e_1\rangle$.

For the variance, we have $2\re[(e^{-i\beta}-1)z]=2\re[e^{-i\beta}z]-2\re z$, and the second term is constant in $\beta$. Writing $z=|z|e^{i\theta}$, we have $\re[e^{-i\beta}z]=|z|\cos(\theta-\beta)$, whose mean over $\beta\sim\Unif[0,2\pi)$ is $0$ and whose mean square is $|z|^2/2$. Hence, the variance is $4\cdot|z|^2/2=2|z|^2$.
\end{proof}


\begin{lemma}\label{lem:phase}
Let $A,B,A',B'\in\C$ and $\vartheta\sim\Unif[0,2\pi)$. Then
\[
\E_\vartheta\Bigl[\,\big|A+(e^{i\vartheta}-1)B\big|^2\,\big|A'+(e^{i\vartheta}-1)B'\big|^2\,\Bigr]\ \ge\ |B|^2|B'|^2 .
\]
\end{lemma}

\begin{proof}
Put $u:=e^{i\vartheta}$, $\tilde A:=A-B$, and $\tilde A':=A'-B'$, so that the two factors are $|\tilde A+uB|^2$ and $|\tilde A'+uB'|^2$. Expanding, we obtain
\[
|\tilde A+uB|^2=|\tilde A|^2+|B|^2+\overline{\tilde A}Bu+\tilde A\overline{B}\,\bar u,
\]
and similarly for the primed quantities. The product is a trigonometric polynomial in $\vartheta$ whose mean is the sum of the terms of total degree zero in $u$, namely
\[
\E_\vartheta\bigl[|\tilde A+uB|^2|\tilde A'+uB'|^2\bigr]=(|\tilde A|^2+|B|^2)(|\tilde A'|^2+|B'|^2)+2\re\bigl[\overline{\tilde A}\tilde A'B\overline{B'}\,\bigr].
\]
Using $\re\bigl[\overline{\tilde A}\tilde A'B\overline{B'}\,\bigr]\ge-|\tilde A||\tilde A'||B||B'|$ and expanding the product of sums, we obtain for the right-hand side above:
\begin{align*}
(|\tilde A|^2+|B|^2)(|\tilde A'|^2+|B'|^2)+2\re\bigl[\overline{\tilde A}\tilde A'B\overline{B'}\,\bigr]
&\ge\ |\tilde A|^2|\tilde A'|^2+|B|^2|B'|^2+\bigl(|\tilde A||B'|-|B||\tilde A'|\bigr)^2\ \\
&\ge\ |B|^2|B'|^2 .
\end{align*}
This proves the lemma.
\end{proof}

\begin{lemma}\label{lem:populations}
Assume the lattice condition and the parameter law of Definition~\ref{def:law}. For $k\ge0$ and $1\le j\le d$, let $q^{(k)}_j:=\E\,|\langle\xi_j|\psi_k\rangle|^2$ and form the vector $q^{(k)}:=(q^{(k)}_j)_j$. Then the following hold.
\begin{enumerate}
\item We have $q^{(0)}=c^2$ and $q^{(k)}=M_1q^{(k-1)}$ for $k\ge1$, where $M_1$ is from Definition~\ref{def:data}.
\item For $k\ge1$, we have $\E[|o_k|^2]=\bar q_{k-1}$, where $\bar q_{k}:=(c^2)^\top q^{(k)}=(c^2)^\top M_1^{k}c^2$, and $o_k$ is defined in Eq.~\eqref{eq:phikok}.
\item The matrix $M_1$ is positive semidefinite, and $\bar q_k\ge1/d$ for every $k\ge0$.
\end{enumerate}
\end{lemma}

\begin{proof}
(1) The case $k=0$ is $|\langle\xi_j|\xi\rangle|^2=c_j^2$. Let $k\ge1$, condition on $|\psi_{k-1}\rangle$, and write $x_j:=\langle\xi_j|\psi_{k-1}\rangle$. Then $\langle\xi_j|\varphi_k\rangle=e^{-i\gamma_k\lambda_j}x_j$ and $o_k=\sum_lc_le^{-i\gamma_k\lambda_l}x_l$. With $w:=e^{-i\beta_k}-1$, we have
\[
\langle\xi_j|\psi_k\rangle=\langle\xi_j|\varphi_k\rangle+w\,o_k\,c_j .
\]
Averaging over $\beta_k$ first and using the identity $\E_\beta|a+wb|^2=|a|^2-2\re(a\bar b)+2|b|^2$ for fixed $a,b\in\C$, we obtain
\begin{equation*}
\E_{\beta_k}|\langle\xi_j|\psi_k\rangle|^2=|x_j|^2-2\re\Bigl[e^{-i\gamma_k\lambda_j}x_j\,c_j\sum_{l=1}^d c_le^{i\gamma_k\lambda_l}\bar x_l\Bigr]+2c_j^2\Big|\sum_{l=1}^d c_le^{-i\gamma_k\lambda_l}x_l\Big|^2 .
\end{equation*}
We now average over $\gamma_k\sim\Unif[0,2\pi/\eta)$. For $j\ne l$, the difference $\lambda_j-\lambda_l$ is a nonzero element of $\eta\Z$, so the average of $e^{-i\gamma_k(\lambda_j-\lambda_l)}$ over $\gamma_k$ vanishes. Hence, the middle term averages to $-2c_j^2|x_j|^2$ and the last term to $2c_j^2\sum_lc_l^2|x_l|^2$. Taking the expectation over $|\psi_{k-1}\rangle$, we obtain
\[
q^{(k)}_j=(1-2c_j^2)\,q^{(k-1)}_j+2c_j^2\sum_{l=1}^d c_l^2q^{(k-1)}_l=\bigl((I-2W)q^{(k-1)}\bigr)_j ,
\]
which is the claimed recursion.

(2) By the same averaging over $\gamma_k$, we have $\E_{\gamma_k}|o_k|^2=\sum_jc_j^2|x_j|^2$, and taking the expectation over $|\psi_{k-1}\rangle$ gives $\E|o_k|^2=\sum_jc_j^2q^{(k-1)}_j=\bar q_{k-1}$. The formula $\bar q_k=(c^2)^\top M_1^kc^2$ follows from (1).

(3) We first show $W\preceq\frac12I$. Let $x\in\R^d$ be a unit vector. Then $x^\top Wx$ is the variance of the coordinates of $x$ under the probability weights $c_j^2$, and these coordinates lie in the interval $[\min_jx_j,\max_jx_j]$. By Popoviciu's inequality on variances \cite{BD00}, we have $x^\top Wx\le\frac14(\max_jx_j-\min_jx_j)^2$. Moreover, we have $(\max_jx_j-\min_jx_j)^2\le2\bigl((\max_jx_j)^2+(\min_jx_j)^2\bigr)\le2\|x\|^2=2$. Hence, we have $x^\top Wx\le\frac12$, so $M_1=I-2W\succeq0$.

Since $M_1$ is symmetric with $0\preceq M_1\preceq I$ (the upper bound because $W\succeq0$), it has an orthonormal eigenbasis $(\phi_i)_i$ with eigenvalues $\mu_i\in[0,1]$. Moreover, the vector $\phi_*:=d^{-1/2}\one$ is an eigenvector with eigenvalue $1$ because $M_1\one=\one$. Therefore,
\[
\bar q_k=\sum_{i=1}^d \mu_i^k\,\langle c^2,\phi_i\rangle^2\ \ge\ \langle c^2,\phi_*\rangle^2=\frac{1}{d}\Big(\sum_{j=1}^d c_j^2\Big)^2=\frac1d .
\]
This completes the proof.
\end{proof}

\subsection{Depth one}\label{sec:onelayer}

In this subsection, we bound the depth-one variance in terms of $M$ from Definition~\ref{def:charfun} and give an explicit lower bound for $M$.

\begin{proposition}\label{prop:onelayer}
Assume the lattice condition with spacing $\eta$ and the parameter law of Definition~\ref{def:law}. Then
\[
\Var_1\ \ge\ M\ \ge\ \frac{49}{384\pi}\,\frac{\eta\,b_0^4}{L^3}.
\]
\end{proposition}

\begin{proof}
We first show $\Var_1\ge M$. By the law of total variance \cite{Dur19}, we have $\Var_1\ge\E_{\gamma_1}\bigl[\Var_{\beta_1}[\ell_1]\bigr]$. Apply Lemma~\ref{lem:lastmixer} with $|\varphi\rangle=|\varphi_1\rangle=Z(\gamma_1)|\xi\rangle$. By the definitions of \(s\), \(t\), and \(|e_1\rangle\) (see Definitions \ref{def:data}, \ref{def:charfun}), we have
\[
\langle\xi|\varphi_1\rangle=s(\gamma_1), \qquad
b_0\langle e_1|\varphi_1\rangle=\langle(H-\ell_0)\xi|Z(\gamma_1)|\xi\rangle=t(\gamma_1). 
\]
Hence, we have $\Var_{\beta_1}[\ell_1]=2|s(\gamma_1)|^2|t(\gamma_1)|^2$, and averaging over $\gamma_1$ gives $\Var_1\ge M$.

We now bound $M$. Put 
\[
\tilde s(\gamma):=e^{i\gamma \ell_0}s(\gamma)=\sum_{j=1}^d c_j^2e^{-i\gamma(\lambda_j-\ell_0)},\qquad
\tilde t(\gamma):=e^{i\gamma \ell_0}t(\gamma),
\]
so that $|\tilde s|=|s|$ and $|\tilde t|=|t|$. For real $x$, we have $\cos x\ge1-x^2/2$, hence
\[
|\tilde s(\gamma)|\ \ge\ \re\tilde s(\gamma)=\sum_{j=1}^d c_j^2\cos\bigl(\gamma(\lambda_j-\ell_0)\bigr)\ \ge\ 1-\frac{\gamma^2b_0^2}{2}.
\]
The values $\lambda_j$ lie in an interval of length $L$, so by Popoviciu's inequality on variances \cite{BD00}, we have $b_0^2\le L^2/4$. Consequently, we have $|s(\gamma)|\ge7/8$ for $|\gamma|\le1/L$.

Next, we have 
\[
\tilde t(0)=\sum_{j=1}^d c_j^2(\lambda_j-\ell_0)=0, \qquad
\tilde t'(0)=-i\sum_{j=1}^d c_j^2(\lambda_j-\ell_0)^2=-ib_0^2. 
\]
Using $|e^{-ix}-1+ix|\le x^2/2$ for real $x$, we derive
\begin{align*}
\big|\tilde t(\gamma)+i\gamma b_0^2\big| &=\Big|\sum_{j=1}^d c_j^2(\lambda_j-\ell_0)\bigl(e^{-i\gamma(\lambda_j-\ell_0)}-1+i\gamma(\lambda_j-\ell_0)\bigr)\Big|\\ 
&\le\ \frac{\gamma^2}{2}\sum_{j=1}^d c_j^2|\lambda_j-\ell_0|^3\ \le\ \frac{\gamma^2}{2}\,L\,b_0^2 ,
\end{align*}
where we used $|\lambda_j-\ell_0|\le L$. Therefore, 
\[
|t(\gamma)|\ge|\gamma|b_0^2-\gamma^2Lb_0^2/2\ge|\gamma|b_0^2/2 \qquad\text{for}\qquad |\gamma|\le1/L.
\]

Under the lattice condition, we have $e^{-i(\gamma+2\pi/\eta)\lambda_j}=e^{-2\pi i\lambda_1/\eta}e^{-i\gamma\lambda_j}$ for every $j$, so $|s|$ and $|t|$ are $2\pi/\eta$-periodic. Hence, the expectation defining $M$ can be taken over $\gamma\sim\Unif[-\pi/\eta,\pi/\eta)$, whose density is $\eta/(2\pi)$. Since $d\ge2$, we have $L\ge\eta$ and the window $\{|\gamma|\le1/L\}$ is contained in $[-\pi/\eta,\pi/\eta)$. Restricting the expectation to this window and using the two bounds above, we obtain
\[
M\ \ge\ \frac{\eta}{2\pi}\int_{-1/L}^{1/L}2\cdot\Big(\frac78\Big)^2\cdot\frac{\gamma^2b_0^4}{4}\,d\gamma=\frac{\eta}{2\pi}\cdot\frac{49\,b_0^4}{128}\cdot\frac{2}{3L^3}=\frac{49}{384\pi}\,\frac{\eta\,b_0^4}{L^3}\,.
\]
This proves the proposition.
\end{proof}

\subsection{\texorpdfstring{Proof of Theorem~\ref{thm:intro-variance}}{Proof of Theorem 1.1}}\label{sec:main}

In this subsection, we prove the depth-independent variance bound, using the mean return probability $\bar q_{p-2}$ of Lemma~\ref{lem:populations}.
Recall $M$ from Definition~\ref{def:charfun}.

\begin{theorem}\label{thm:main}
Let $V$, $H$, and $|\xi\rangle$ be as in Section~\ref{sec:setting}, with $d\ge2$. Assume the lattice condition with spacing $\eta$ and the parameter law of Definition~\ref{def:law}. Then 
we have
\[
\Var_1\ \ge\ M,\qquad \Var_p\ \ge\ \E\,|o_{p-1}|^4\cdot M\ \ge\ \bar q_{p-2}^{\,2}\,M\ \ge\ \frac{M}{d^2}\quad(p\ge2).
\]
In particular, for every $p\ge1$, we have
\[
\Var_p\ \ge\ \frac{49}{384\pi}\,\frac{\eta\,b_0^4}{d^2L^3}.
\]
\end{theorem}

\begin{proof}
The case $p=1$ is Proposition~\ref{prop:onelayer}. Let $p\ge2$, and denote by $\mathcal F$ the collection of parameters $(\gamma_1,\dots,\gamma_{p-1},\beta_1,\dots,\beta_{p-2})$. The state $|\varphi_{p-1}\rangle$ and the overlap $o_{p-1}$ are functions of $\mathcal F$, while $\beta_{p-1}$, $\gamma_p$, and $\beta_p$ are independent of $\mathcal F$ and of each other.
By the law of total variance and Lemma~\ref{lem:lastmixer} applied to $|\varphi\rangle=|\varphi_p\rangle$, we have
\[
\Var_p\ \ge\ \E\bigl[\Var_{\beta_p}[\ell_p]\bigr]=2b_0^2\,\E\Big[|\langle\xi|\varphi_p\rangle|^2\,|\langle e_1|\varphi_p\rangle|^2\Big],
\]
where the outer expectation is over $\mathcal F$, $\beta_{p-1}$, and $\gamma_p$.

Since $|\varphi_p\rangle=Z(\gamma_p)|\psi_{p-1}\rangle$ and
\[
|\psi_{p-1}\rangle=R(\beta_{p-1})|\varphi_{p-1}\rangle=|\varphi_{p-1}\rangle+(e^{-i\beta_{p-1}}-1)\,o_{p-1}|\xi\rangle,
\]
we have
\[
\langle\xi|\varphi_p\rangle=A+(e^{-i\beta_{p-1}}-1)B,\qquad \langle e_1|\varphi_p\rangle=A'+(e^{-i\beta_{p-1}}-1)B',
\]
with
\[
A:=\langle\xi|Z(\gamma_p)|\varphi_{p-1}\rangle,\quad A':=\langle e_1|Z(\gamma_p)|\varphi_{p-1}\rangle,\quad B:=o_{p-1}\,s(\gamma_p),\quad B':=o_{p-1}\,\langle e_1|Z(\gamma_p)|\xi\rangle .
\]
None of $A,A',B,B'$ depends on $\beta_{p-1}$. Conditioning on $\mathcal F$ and $\gamma_p$ and applying Lemma~\ref{lem:phase} with $\vartheta=-\beta_{p-1}$, we obtain
\[
\E_{\beta_{p-1}}\Big[|\langle\xi|\varphi_p\rangle|^2|\langle e_1|\varphi_p\rangle|^2\Big]\ \ge\ |B|^2|B'|^2=|o_{p-1}|^4\,|s(\gamma_p)|^2\,|\langle e_1|Z(\gamma_p)|\xi\rangle|^2 .
\]

Since $b_0\langle e_1|Z(\gamma_p)|\xi\rangle=t(\gamma_p)$ and $\gamma_p$ is independent of $\mathcal F$, we take the expectation over $\gamma_p$ and then over $\mathcal F$ in the preceding bound. Inserting the result into the variance estimate gives
\[
\Var_p\ \ge\ \E\bigl[|o_{p-1}|^4\bigr]\cdot2\,\E_{\gamma\sim\Unif[0,2\pi/\eta)}\bigl[|s(\gamma)|^2|t(\gamma)|^2\bigr]=\E\bigl[|o_{p-1}|^4\bigr]\cdot M .
\]
By Jensen's inequality \cite{Dur19}, we have
\[
\E\bigl[|o_{p-1}|^4\bigr]\ge\bigl(\E\bigl[|o_{p-1}|^2\bigr]\bigr)^2 .
\]
By Lemma~\ref{lem:populations}, the second moment is equal to $\E[|o_{p-1}|^2]=\bar q_{p-2}$, and $\bar q_{p-2}\ge1/d$. 

Combining the above inequalities with Proposition~\ref{prop:onelayer} proves the claims for $p\ge2$. For $p=1$, we use $\Var_1\ge M\ge M/d^2$.
\end{proof}

The last assertion of Theorem~\ref{thm:main} is precisely Theorem~\ref{thm:intro-variance}. We remark that the constants in Theorem~\ref{thm:main} have not been optimized. The intermediate bound $\Var_p\ge\E[|o_{p-1}|^4] M$ is the sharpest statement given by the proof, and $\E[|o_{p-1}|^4]$ can be computed from the parameter law for any given instance.

\section{\texorpdfstring{Applications of Theorem~\ref{thm:intro-variance}}{Applications of Theorem 1.1}}\label{sec:applications}

In this section, we first derive the corresponding variance bound for the final mixer gradient and state the scaling condition for a family of instances. We then specialize the bounds to integer-valued local objective functions in Subsection~\ref{sec:local} and to MaxCut in Subsection~\ref{sec:maxcut}.

\subsection{Variance of the derivative of the loss}\label{sec:vargrad}
We continue to use the notation and assumptions of Section~\ref{sec:setting}. For $p\ge1$, we consider the partial derivative of the loss function with respect to the last mixer angle,
\[
 g_p:=\frac{\partial\ell_p}{\partial\beta_p}.
\]

\begin{corollary}\label{cor:gradient}
Under the assumptions of Theorem~\ref{thm:main}, we have
\[
 \Var[g_1]=M,
\]
and, for $p\ge2$,
\[
 \Var[g_p]
 \ge \E\bigl[|o_{p-1}|^4\bigr] M
 \ge \bar q_{p-2}^{\,2}M
 \ge \frac{M}{d^2}.
\]
In particular, for every $p\ge1$, we have
\[
 \Var[g_p]
 \ge
 \frac{49}{384\pi}\,\frac{\eta b_0^4}{d^2L^3}.
\]
\end{corollary}

\begin{proof}
Condition on all parameters except $\beta_p$ and put
\[
 z_p:=b_0\langle\xi|\varphi_p\rangle\langle\varphi_p|e_1\rangle.
\]
By Lemma~\ref{lem:lastmixer}, we have
\[
 \ell_p=C_p+2\re\bigl[(e^{-i\beta_p}-1)z_p\bigr],
 \qquad
 g_p=2\re\bigl[-ie^{-i\beta_p}z_p\bigr],
\]
where $C_p$ and $z_p$ do not depend on $\beta_p$. Therefore,
\[
 \E_{\beta_p}[g_p]=0,
 \qquad
 \E_{\beta_p}[g_p^2]=2|z_p|^2
 =\Var_{\beta_p}[\ell_p].
\]
The argument in the proof of Theorem~\ref{thm:main}, beginning with Lemma~\ref{lem:lastmixer}, applies without change to \(\Var[g_p]\). For $p=1$, averaging $2|z_1|^2$ over $\gamma_1$ gives $M$. For $p\ge2$, the remaining inequalities are exactly those in Theorem~\ref{thm:main}.
\end{proof}

\begin{remark}\label{rem:family-scaling}
The above lower bounds are instance dependent. For a family $(H_n,|\xi_n\rangle)$ and the normalized loss $\widehat\ell_{p,n}:=\ell_{p,n}/\|H_n\|$, Corollary~\ref{cor:gradient} gives
\[
 \Var\Bigl[
 \frac{\partial\widehat\ell_{p,n}}{\partial\beta_p}
 \Bigr]
 \ge
 \frac{49}{384\pi}\,
 \frac{\eta_n b_{0,n}^4}
 {d_n^2L_n^3\|H_n\|^2}.
\]
Thus, an inverse-polynomial conclusion for a growing family requires the quantity on the right to be bounded below by an inverse polynomial in the problem size. The lattice condition alone does not imply this. For example, on $n$ qubits, take
\[
 |\xi_n\rangle=|+\rangle^{\otimes n},
 \qquad
 H_n=|1^n\rangle\langle1^n|,
 \qquad q_n:=2^{-n}.
\]
Then we have $d=2$ and $L=\eta=\|H_n\|=1$, but a direct computation of $M$ from Definition~\ref{def:charfun} gives
\[
 \Var\Bigl[
 \frac{\partial\ell_{1,n}}{\partial\beta_1}
 \Bigr]
 =4q_n^2(1-q_n)^2\bigl(1-3q_n(1-q_n)\bigr)
 =\Theta(4^{-n}).
\]
\end{remark}

\subsection{Local objective functions}\label{sec:local}

In this subsection, we specialize Theorem~\ref{thm:intro-variance} to integer-valued local objective functions and compute the constants. The result is a finite-depth counterpart of Theorem~IV.4 of \cite{TNB25}. Notice that the lattice condition with a fixed $\eta$ can be reduced to the integer-valued case by rescaling the Hamiltonian $H$.

Let $n\ge1$, let $\Bool^n:=\{0,1\}^n$, and $s\ge1$ be an integer. An objective function $F\colon\Bool^n\to\R$ is called \emph{$s$-local} if
\begin{equation}\label{eq:slocal}
F(x)=\sum_{j=1}^Tf_j(x_{Q_j}),\qquad x\in\Bool^n,
\end{equation}
where $Q_1,\dots,Q_T\subseteq\{1,\dots,n\}$ are distinct sets with $|Q_j|=s$, $x_{Q_j}:=(x_i)_{i\in Q_j}\in\{0,1\}^{Q_j}$, and $f_j\colon\{0,1\}^{Q_j}\to\R$. We call $T$ the number of terms of $F$, and put
\[
K:=\max_{1\le j\le T}\ \max_{z\in\{0,1\}^{Q_j}}|f_j(z)|,\qquad
\lambda_{\max}:=\max_{x\in\Bool^n}|F(x)|.
\]
Consider the Hamiltonian
\[
H:=\sum_{x\in\Bool^n}F(x)|x\rangle\langle x|
\]
on the Hilbert space $V=(\C^2)^{\otimes n}$, which is diagonal in the computational basis $\{|x\rangle\}$ of $V$ with eigenvalues given by $F(x)$.

Let $|\xi\rangle\in V$ be a fixed unit vector, and $X$ be a random element of $\Bool^n$ with $\Pr[X=x]=|\langle x|\xi\rangle|^2$. In the notation of Section~\ref{sec:setting}, the eigenvalues $\lambda_1<\dots<\lambda_d$ are the values of $F$ attained with positive probability by $F(X)$, and $c_j^2=\Pr[F(X)=\lambda_j]$. Consequently, we have $\ell_0=\E[F(X)]$ and $b_0^2=\Var[F(X)]$. The state $|\xi\rangle$ is an eigenvector of $H$ if and only if $d=1$. The locality parameter $s$ should not be confused with the function $s(\gamma)$ of Definition~\ref{def:charfun}, which always appears with its argument.

\begin{theorem}\label{thm:local}
Let $F\colon\Bool^n\to\Z$ be $s$-local as in \eqref{eq:slocal} with $T$ terms, and let $H$, $|\xi\rangle$, $K$, $\lambda_{\max}$, and $b_0$ be as above. Assume that $|\xi\rangle$ is not an eigenvector of $H$, and let the parameters be sampled as in Definition~\ref{def:law} with $\eta=1$. Put $\hat\ell_p:=\ell_p/\lambda_{\max}$. Then the following hold for every depth $p\ge1$.
\begin{enumerate}
\item We have
\[
\Var[\ell_p]\ \ge\ \frac{49}{384\pi}\,\frac{b_0^4}{(2KT+1)^2(2KT)^3}
\qquad\text{and}\qquad
\Var[\hat\ell_p]\ \ge\ \frac{49}{27648\pi}\,\frac{b_0^4}{(KT)^7}.
\]
\item If every $f_j$ takes values in $[0,K]$, then
\[
\Var[\ell_p]\ \ge\ \frac{49}{384\pi}\,\frac{b_0^4}{(KT+1)^2(KT)^3}
\qquad\text{and}\qquad
\Var[\hat\ell_p]\ \ge\ \frac{49}{1536\pi}\,\frac{b_0^4}{(KT)^7}.
\]
\item If $|\xi\rangle=|+\rangle^{\otimes n}$, then $b_0^2\ge4^{-s}$, and consequently
\[
\Var[\hat\ell_p]\ \ge\ \frac{49}{27648\pi}\,\frac{1}{16^{s}\,(KT)^7}.
\]
\end{enumerate}
The same lower bounds hold with $\ell_p$ and $\hat\ell_p$ replaced by $\partial\ell_p/\partial\beta_p$ and $\partial\hat\ell_p/\partial\beta_p$, respectively.
\end{theorem}

\begin{proof}
As $F$ takes integer values, the lattice condition holds with $\eta=1$, and since $|\xi\rangle$ is not an eigenvector of $H$, we have $d\ge2$. Hence, Theorem~\ref{thm:intro-variance} and Corollary~\ref{cor:gradient} apply. The gradient bounds follow from Corollary~\ref{cor:gradient} in the same way as the bounds for \(\ell_p\) and \(\hat\ell_p\) follow from Theorem~\ref{thm:intro-variance}. It remains to bound $d$, $L$, $\lambda_{\max}$, and $b_0$.

(1) By Eq.~\eqref{eq:slocal} and the definition of $K$, we have $|F(x)|\le KT$ for all $x\in\Bool^n$. Hence, every $\lambda_j$ lies in $[-KT,KT]$, so $L\le2KT$ and $\lambda_{\max}\le KT$. The $d$ values $\lambda_j$ are distinct integers in an interval of length $L$, so $d\le L+1\le2KT+1$. Inserting $\eta=1$ and these bounds into Theorem~\ref{thm:intro-variance} gives the first inequality. Since $d\ge2$, the function $F$ takes at least two distinct integer values, one of which is nonzero, so $\lambda_{\max}\ge1$. Therefore, we have $KT\ge1$. Since \(\Var[\hat\ell_p]=\lambda_{\max}^{-2}\Var[\ell_p]\ge (KT)^{-2}\Var[\ell_p]\) and \(2KT+1\le 3KT\), the second inequality follows from the first.

(2) If every $f_j$ takes values in $[0,K]$, then $0\le F(x)\le KT$ for all $x$, so $L\le KT$, $d\le KT+1$, and $\lambda_{\max}\le KT$. Theorem~\ref{thm:intro-variance} gives the first inequality. Since \(KT\ge1\), we have \((KT+1)^2\le4(KT)^2\), and the second inequality follows as in part (1).

(3) Let $|\xi\rangle=|+\rangle^{\otimes n}$, so that $X$ is uniformly distributed on $\Bool^n$. For $U\subseteq\{1,\dots,n\}$, let $x^U:=\prod_{i\in U}x_i$ and let $\one_U\in\Bool^n$ be the indicator vector of $U$. The $2^n$ monomials $x^U$ form a basis of the space of functions on $\Bool^n$. Indeed, we have $x^U(\one_V)=1$ if $U\subseteq V$ and $x^U(\one_V)=0$ otherwise, so the matrix $\bigl(x^U(\one_V)\bigr)_{U,V}$ is triangular with unit diagonal when the subsets are listed in any order of nondecreasing cardinality. Hence, $F$ has a unique expansion
\begin{equation}\label{eq:monomial}
F(x)=\sum_{U\subseteq\{1,\dots,n\}}a_U\,x^U,\qquad x\in\Bool^n,
\end{equation}
with real coefficients $a_U$. Evaluating \eqref{eq:monomial} at $x=\one_U$ gives $a_U=F(\one_U)-\sum_{V\subsetneq U}a_V$, so by induction on $|U|$ every $a_U$ is an integer. Each term $f_j(x_{Q_j})$ in \eqref{eq:slocal} is a function of the coordinates in $Q_j$ alone, so its expansion involves only monomials $x^U$ with $U\subseteq Q_j$. By the uniqueness of \eqref{eq:monomial}, we have $a_U=0$ whenever $|U|>s$.

Now, put $y_i:=1-2x_i\in\{-1,1\}$ and $\chi_S(y):=\prod_{i\in S}y_i$ for $S\subseteq\{1,\dots,n\}$. Since $x_i=(1-y_i)/2$, we have
\[
x^U=2^{-|U|}\prod_{i\in U}(1-y_i)=2^{-|U|}\sum_{S\subseteq U}(-1)^{|S|}\chi_S(y),
\]
and therefore
\[
F(x)=\sum_{S\subseteq\{1,\dots,n\}}\widehat F(S)\,\chi_S(y),\qquad
\widehat F(S):=(-1)^{|S|}\sum_{U\supseteq S,\ |U|\le s}a_U\,2^{-|U|},
\]
which is the Fourier expansion of $F$ on the Boolean cube \cite[Ch.~1]{ODonnell14}.
Since every $a_U$ is an integer, the number $2^s\widehat F(S)=(-1)^{|S|}\sum_{U\supseteq S,\,|U|\le s}a_U2^{\,s-|U|}$ is an integer for every $S$. When $X$ is uniform on $\Bool^n$, the vector $Y:=(1-2X_i)_i$ is uniform on $\{-1,1\}^n$. Hence, the expectation $\E[\chi_S(Y)\chi_{S'}(Y)]$ is equal to $1$ if $S=S'$ and $0$ otherwise, because for $S\ne S'$ the product $\chi_S\chi_{S'}=\chi_{S\triangle S'}$ has mean zero. Hence, we have $\E\,F(X)=\widehat F(\emptyset)$ and
\[
b_0^2=\Var[F(X)]=\sum_{S\ne\emptyset}\widehat F(S)^2 .
\]
Since $d\ge2$, the function $F$ is not constant, so $\widehat F(S)\ne0$ for some $S\ne\emptyset$. Hence \(|\widehat F(S)|\ge2^{-s}\) for some \(S\ne\emptyset\), and therefore \(b_0^2\ge4^{-s}\). Inserting $b_0^4\ge16^{-s}$ into the second inequality of part (1) proves the displayed bound.
\end{proof}

Note that $T\le\binom ns\le n^s/s!$, and Theorem~\ref{thm:local}(3) gives
\begin{equation}\label{eq:local-n}
\Var[\hat\ell_p]\ \ge\ \frac{49\,(s!)^7}{27648\pi\,16^{s}K^7\,n^{7s}}
\qquad(p\ge1)
\end{equation}
for the uniform-superposition initial state $|\xi\rangle=|+\rangle^{\otimes n}$. For the same variance, Theorem~IV.4 of \cite{TNB25} gives the lower bound
\[
\frac{(s!)^2}{12\,K^2\,n^{2s}}
\]
when the depth $p$ is sufficiently large, without an explicit threshold for $p$. The bound \eqref{eq:local-n} holds at all depths, at the cost of a larger power of $n$.

We note that two types of dependence in Theorem~\ref{thm:local} are unavoidable in the finite-depth setting.  First, the dependence on $s$ in part (3) cannot be removed. For $F(x)=x^S$ with $|S|=s$, we have $T=K=1$ and $b_0^2=2^{-s}(1-2^{-s})$, so $b_0^2$ decays exponentially in $s$. For $s=1$, this value equals $4^{-s}$, so the bound $b_0^2\ge4^{-s}$ is attained. Second, for a general initial state, the factor $b_0^4$ in parts (1) and (2) cannot be replaced by a quantity depending only on $s$, $K$, and $T$. For $F(x)=x_1$ and $q\in(0,1)$, put $|\xi\rangle:=\bigl(\sqrt{1-q}\,|0\rangle+\sqrt q\,|1\rangle\bigr)\otimes|+\rangle^{\otimes(n-1)}$. Here $s=1$ and $K=T=1$. Moreover, we have $d=2$, the carried eigenvalues are $0$ and $1$ with weights $1-q$ and $q$, and $b_0^2=q(1-q)$. The calculation of $M$ in Remark~\ref{rem:family-scaling} depends only on these data and gives $\Var[\partial\ell_1/\partial\beta_1]=4q^2(1-q)^2\bigl(1-3q(1-q)\bigr)$, which tends to $0$ as $q\to0$. The loss itself behaves in the same way. By Lemma~\ref{lem:lastmixer} applied to $|\varphi_1\rangle=Z(\gamma_1)|\xi\rangle$, for which $\langle\varphi_1|H|\varphi_1\rangle=\ell_0$, we have $|\ell_1-\ell_0|\le4b_0$ for every instance, so $\Var[\ell_1]\le16b_0^2=16q(1-q)$ in the present example. In contrast, the bound of Theorem~IV.4 of \cite{TNB25} does not depend on the initial state, because the large-depth variance of \cite{TNB25} depends only on the carried eigenvalues.

\subsection{MaxCut}\label{sec:maxcut}

In this subsection, we apply Theorem~\ref{thm:local} to MaxCut \cite{GW95,FGG14} with the uniform-superposition initial state $|\xi\rangle=|+\rangle^{\otimes n}=2^{-n/2}\sum_x|x\rangle$, and compute the constants.

Let $\Gamma$ be a simple graph (without loops and without multiple edges) with $n$ vertices and $m\ge1$ edges. For an assignment $x\in\{0,1\}^n$, denote by $f(x)$ be the number of cut edges, i.e., the number of edges $\{i,j\}$ whose endpoints satisfy $x_i\ne x_j$. As before, consider the Hamiltonian $H=\sum_xf(x)|x\rangle\langle x|$ acting on the Hilbert space $V=(\C^2)^{\otimes n}$. The eigenspaces of $H$ with $\Pi_\lambda|\xi\rangle\ne0$ are indexed by the cut values $\lambda$ attained by $f$, and $c_\lambda^2=\Pr[f(X)=\lambda]$ for $X$ uniform in $\{0,1\}^n$. For an edge $e=\{i,j\}$, let $I_e(x):=x_i\oplus x_j$, so that $f(x)=\sum_eI_e(x)$. Thus, in the notation of Subsection~\ref{sec:local}, the function $F=f$ is $2$-local with $T=m$ terms, the local terms are the functions $I_e$, each taking values in $\{0,1\}$, and $K=1$.

\begin{corollary}\label{cor:maxcut}
In the setting above, with parameters sampled as in Definition~\ref{def:law} with $\eta=1$, for every depth $p\ge1$, we have
\[
\Var[\ell_p]\ \ge\ \frac{49}{6144\pi}\,\frac{1}{m\,(m+1)^2}\ \ge\ \frac{49}{24576\pi}\,\frac{1}{m^3}\,.
\]
For the normalized loss $\hat\ell_p:=\ell_p/\lambda_{\max}$, where $\lambda_{\max}$ is the maximum cut value, we have
\[
\Var[\hat\ell_p]\ \ge\ \frac{49}{24576\pi}\,\frac{1}{m^5}\,.
\]
{The same two lower bounds hold with $\ell_p$ and $\hat\ell_p$ replaced by $\partial\ell_p/\partial\beta_p$ and $\partial\hat\ell_p/\partial\beta_p$, respectively.}
\end{corollary}

\begin{proof}
Since \(m\ge1\), the function \(f\) is not constant and therefore \(|\xi\rangle\) is not an eigenvector of \(H\). Moreover, \(f\) is integer-valued and \(2\)-local with \(T=m\), \(K=1\), and local terms taking values in \(\{0,1\}\). Thus, Theorem~\ref{thm:local}(2) applies with \(KT=m\), and \(\lambda_{\max}\) is the maximum cut value.

We compute $b_0^2=\Var[f(X)]$. Each $I_e(X)=X_i\oplus X_j$ is a Bernoulli variable with parameter $\frac12$, and hence $\Var[I_e]=\frac14$. For two distinct edges $e=\{i,j\}$ and $e'=\{k,l\}$, since $\Gamma$ is simple, either the edges are vertex-disjoint or they share exactly one vertex. In the first case, the variables $I_e$ and $I_{e'}$ are independent. In the second case, say $j=k$, we have $\Pr[I_e=1,I_{e'}=1]=\Pr[X_i\ne X_j,X_j\ne X_l]=\frac14=\Pr[I_e=1]\Pr[I_{e'}=1]$. In both cases, we have $\operatorname{Cov}[I_e,I_{e'}]=0$. Therefore, we have $b_0^2=m/4$.

Inserting $b_0^4=m^2/16$ and $KT=m$ into the two inequalities of Theorem~\ref{thm:local}(2) leads to
\[
\Var[\ell_p]\ \ge\ \frac{49}{384\pi}\cdot\frac{m^2/16}{(m+1)^2m^3}=\frac{49}{6144\pi}\,\frac{1}{m(m+1)^2}
\]
and
\[
\Var[\hat\ell_p]\ \ge\ \frac{49}{1536\pi}\cdot\frac{m^2/16}{m^7}=\frac{49}{24576\pi}\,\frac{1}{m^5}.
\]
Since $(m+1)^2\le4m^2$ for $m\ge1$, the first right-hand side is at least $\frac{49}{24576\pi}m^{-3}$. The gradient bounds follow from the corresponding assertion of Theorem~\ref{thm:local}.
\end{proof}

Note that the value \(b_0^2=m/4\) calculated in the proof exceeds the general lower bound \(b_0^2\ge4^{-s}=\frac1{16}\) of Theorem~\ref{thm:local}(3) by a factor of \(4m\). Using only the general bound \(b_0^2\ge\frac1{16}\) in Theorem~\ref{thm:local}(2) would give \(\Var[\hat\ell_p]\ge\frac{49}{393216\pi}m^{-7}\). Thus, the exact value of \(b_0^2\) improves the resulting lower bound from order \(m^{-7}\) to order \(m^{-5}\).

\section{\texorpdfstring{Proof of Theorem~\ref{thm:intro-depth}}{Proof of Theorem 1.2}}\label{sec:reach}

{In this section, we prove Theorem~\ref{thm:intro-depth} and develop two further reachability bounds. We first express the prepared state as a linear combination of at most $p+1$ rotated copies of the initial state. This gives the Grover envelope of Theorem~\ref{thm:grover}, the spectrum-dependent refinement of Theorem~\ref{thm:inject}, and the consecutive-run obstruction of Theorem~\ref{thm:run}. Numerical observations are recorded in Subsection~\ref{subsec:numerics}.}

\subsection{Setting for this section}

Throughout this section, the parameters $\beta_k,\gamma_k$ are arbitrary real numbers, not random variables, and the lattice condition of Definition \ref{def:lattice} is assumed only where stated. We fix a target index $j_0\in\{1,\dots,d\}$ and write
\[
|e\rangle:=|\xi_{j_0}\rangle,\qquad \lambda_*:=\lambda_{j_0},\qquad \sin\theta_0:=c_{j_0}\in(0,1),\qquad \theta_0\in(0,\pi/2),
\]
so that $\langle e|\xi\rangle=\sin\theta_0$. The vector $|e\rangle$ should not be confused with $|e_1\rangle$ of Definition~\ref{def:data}, which is not used in this section. Since $\sum_jc_j^2=1$ and $d\ge2$, we have $c_{j_0}<1$, so $\cos\theta_0>0$ and the unit vector
\[
|\xi'\rangle:=\frac{|\xi\rangle-\sin\theta_0\,|e\rangle}{\cos\theta_0}\in e^\perp
\]
is well defined, where $e^\perp$ denotes the orthogonal complement of $|e\rangle$ in $W_0$. We write $P_\perp$ for the orthogonal projection onto $e^\perp$ and $P_S$ for the orthogonal projection onto $\operatorname{span}\{|\xi_j\rangle:\lambda_j\in S\}$ for a set $S$ of real numbers.

Replacing $H$ by $H-\lambda_*$ multiplies $Z(\gamma)$ by the scalar $e^{i\gamma\lambda_*}$, which commutes with every $R(\beta)$, and therefore multiplies every $|\psi_k\rangle$ by a scalar of modulus one. The overlaps of $|\psi_k\rangle$ with fixed vectors change only by a phase. {In Theorem~\ref{thm:inject}, we therefore assume $\lambda_*=0$, so that $Z(\gamma)|e\rangle=|e\rangle$ for all $\gamma$ and $\langle e|Z(\gamma)x\rangle=\langle e|x\rangle$ for all $x\in W_0$.} For $\gamma\in\R$, we put
\[
\sigma(\gamma)^2:=\frac{1}{\cos^2\theta_0}\sum_{j\ne j_0}c_j^2\sin^2\frac{\gamma(\lambda_j-\lambda_*)}{2},\qquad \sigma_{\max}:=\sup_{\gamma\in\R}\sigma(\gamma).
\]
Since $\sum_{j\ne j_0}c_j^2=\cos^2\theta_0$, we have $0\le\sigma(\gamma)\le1$.

\subsection{The prepared state as a linear combination of rotated initial states}
Recall the notation \eqref{eq:phikok}.

\begin{lemma}\label{lem:inject}
For $1\le k\le p$, let $w_k:=e^{-i\beta_k}-1$, and let $S_k:=\gamma_{k+1}+\dots+\gamma_p$ for $0\le k\le p$, so that $S_p=0$. Then
\[
|\psi_p\rangle=\sum_{k=0}^{p}h_k\,Z(S_k)|\xi\rangle,\qquad h_0:=1,\qquad h_k:=w_ko_k\quad(1\le k\le p),
\]
and $|h_k|\le2|o_k|\le2$ for $1\le k\le p$.
\end{lemma}

\begin{proof}
We prove by induction on $k$ that $|\psi_k\rangle=\sum_{l=0}^{k}h_lZ(\gamma_{l+1}+\dots+\gamma_k)|\xi\rangle$ with the stated $h_l$. For $k=0$, this is $|\psi_0\rangle=|\xi\rangle$. Let $k\ge1$. Since $R(\beta_k)|\varphi_k\rangle=|\varphi_k\rangle+w_ko_k|\xi\rangle$ and $|\varphi_k\rangle=Z(\gamma_k)|\psi_{k-1}\rangle$, the induction hypothesis gives
\[
|\psi_k\rangle = R(\beta_k)|\varphi_k\rangle = \sum_{l=0}^{k-1}h_lZ(\gamma_{l+1}+\dots+\gamma_k)|\xi\rangle+h_kZ(0)|\xi\rangle ,
\]
which is the claim for $k$. The bound $|h_k|\le2|o_k|\le2$ follows from the inequalities $|w_k|\le2$ and $|o_k|\le\|\xi\|\|\varphi_k\|=1$.
\end{proof}

In the basis $\{|\xi_j\rangle\}$ of $W_0$, Lemma~\ref{lem:inject} reads $\langle\xi_j|\psi_p\rangle=c_j\,E(\lambda_j)$ with
\begin{equation}\label{eq:expsum}
E(\lambda):=\sum_{k=0}^ph_ke^{-iS_k\lambda},\qquad\lambda\in\R .
\end{equation}
Thus, $\langle\xi_j|\psi_p\rangle=c_jE(\lambda_j)$, where $E$ is an exponential sum with at most $p+1$ terms.
The layer-by-layer recursion \(|\psi_k\rangle=|\varphi_k\rangle+w_ko_k|\xi\rangle\) used in the proof is the cost-resolved amplitude recursion used by Kiktenko, Krendeleva, and Fedorov \cite{KKF26}. For the uniform feasible-state initialization, B\"artschi and
Eidenbenz \cite{BE20} observed that GM-QAOA assigns equal amplitudes to
feasible states with the same cost value, so the dynamics depends only
on the cost values and their initial weights.

\subsection{The Grover envelope}

\begin{theorem}\label{thm:grover}
Let $p\ge1$ and let $\beta_1,\dots,\beta_p,\gamma_1,\dots,\gamma_p\in\R$ be arbitrary. Then the following hold.
\begin{enumerate}
\item For $0\le k\le p$, define $\phi_k\in[0,\pi/2]$ by $\sin\phi_k:=|\langle e|\psi_k\rangle|$. Then $\phi_0=\theta_0$ and $\phi_k\le\min\{\phi_{k-1}+2\theta_0,\ \pi/2\}$ for $1\le k\le p$. Consequently,
\[
|\langle e|\psi_p\rangle|\ \le\ \sin\bigl(\min\{(2p+1)\theta_0,\ \pi/2\}\bigr).
\]
\item For every unit vector $|\chi\rangle\in W_0$, we have $|\langle\chi|\psi_p\rangle|\le(2p+1)\sup_{\gamma\in\R}|\langle\chi|Z(\gamma)|\xi\rangle|$.
\item For every set $S\subset\R$, we have $\|P_S\psi_p\|^2\le(2p+1)^2\,\|P_S\xi\|^2$.
\end{enumerate}
\end{theorem}

\begin{proof}
We prove parts (2) and (3) first and part (1) last.

(2) By Lemma~\ref{lem:inject}, we have $|\langle\chi|\psi_p\rangle|\le\sum_{k=0}^p|h_k|\,|\langle\chi|Z(S_k)|\xi\rangle|\le(1+2p)\sup_\gamma|\langle\chi|Z(\gamma)|\xi\rangle|$.

(3) Since $P_S$ commutes with $Z(\gamma)$ and $Z(\gamma)$ is unitary, we have $\|P_SZ(\gamma)\xi\|=\|Z(\gamma)P_S\xi\|=\|P_S\xi\|$ for every $\gamma$. By Lemma~\ref{lem:inject} and the triangle inequality, we have $\|P_S\psi_p\|\le\sum_k|h_k|\,\|P_S\xi\|\le(2p+1)\|P_S\xi\|$.

(1) The equality $\phi_0=\theta_0$ is the definition of $\theta_0$. Fix $1\le k\le p$ and write $\beta:=\beta_k$, $|\Phi\rangle:=|\varphi_k\rangle$, and $\phi:=\phi_{k-1}$ for short. Since $Z(\gamma_k)$ is unitary and commutes with the projection onto $|e\rangle$, we have $|\langle e|\Phi\rangle|=|\langle e|\psi_{k-1}\rangle|=\sin\phi$. Multiplying $|\Phi\rangle$ by a phase changes neither $|\langle e|\Phi\rangle|$ nor $|\langle e|R(\beta)\Phi\rangle|$, so we can assume $x:=\langle e|\Phi\rangle=\sin\phi\ge0$. 

Introduce the shortcut notation 
\[
r:=\cos\phi, \qquad \mathsf{s}:=\sin\theta_0, \qquad \mathsf{c}:=\cos\theta_0, \qquad y:=\langle\xi'|\Phi\rangle. 
\]
Since $|\xi'\rangle\in e^\perp$ is a unit vector and $\|P_\perp\Phi\|^2=1-x^2=r^2$, we have $|y|\le r$. From $|\xi\rangle=\mathsf{s}|e\rangle+\mathsf{c}|\xi'\rangle$, we get $\langle\xi|\Phi\rangle=\mathsf{s}x+\mathsf{c}y$. From $R(\beta)|\Phi\rangle=|\Phi\rangle+w\langle\xi|\Phi\rangle|\xi\rangle$ with $w:=e^{-i\beta}-1$, we obtain
\[
\langle e|R(\beta)\Phi\rangle=x+wz,\qquad z:=\mathsf{s}(\mathsf{s}x+\mathsf{c}y).
\]
Since $x+wz=(x-z)+e^{-i\beta}z$, the maximum of $|x+wz|$ over $\beta\in\R$ is equal to $|x-z|+|z|$, and $x-z=\mathsf{c}(\mathsf{c}x-\mathsf{s}y)$. Hence, we have
\[
\sin\phi_k=|\langle e|\psi_k\rangle|\ \le\ F(y):=\mathsf{c}|\mathsf{c}x-\mathsf{s}y|+\mathsf{s}|\mathsf{s}x+\mathsf{c}y| .
\]
The function $F$ is convex on $\C$, so its maximum over the disc $\{|y|\le r\}$ is attained on the circle $|y|=r$. Write $y=re^{i\omega}$ and $\mathsf{t}:=\cos\omega\in[-1,1]$. Then
\[
F(y)=\mathsf{c}\sqrt{A-B\mathsf{t}}+\mathsf{s}\sqrt{C+B\mathsf{t}},\qquad A:=\mathsf{c}^2x^2+\mathsf{s}^2r^2,\quad C:=\mathsf{s}^2x^2+\mathsf{c}^2r^2,\quad B:=2\mathsf{c}\mathsf{s}xr .
\]
Let $g(\mathsf{t})$ denote the right-hand side, defined on the interval $J:=\{\mathsf{t}\in\R:A-B\mathsf{t}\ge0,\ C+B\mathsf{t}\ge0\}\supset[-1,1]$. The function $g$ is concave on $J$, being a sum of concave functions. 

If $B=0$, then $x=0$ or $r=0$, which means $\phi\in\{0,\pi/2\}$. If $\phi=\pi/2$, there is nothing to prove, and if $\phi=0$, then $r=1$ and $g\equiv\mathsf{c}\mathsf{s}+\mathsf{s}\mathsf{c}=\sin2\theta_0=\sin(\phi+2\theta_0)$. 

Suppose that $B>0$. On the interior of $J$, we have $g'(\mathsf{t})=0$ if and only if $\mathsf{c}\sqrt{C+B\mathsf{t}}=\mathsf{s}\sqrt{A-B\mathsf{t}}$, that is, $\mathsf{c}^2(C+B\mathsf{t})=\mathsf{s}^2(A-B\mathsf{t})$. Using $\mathsf{c}^2C-\mathsf{s}^2A=(\mathsf{c}^4-\mathsf{s}^4)r^2=(\mathsf{c}^2-\mathsf{s}^2)r^2=r^2\cos2\theta_0$ and $B=xr\sin2\theta_0$, the unique solution is
\[
\mathsf{t}^*=-\frac{r\cos2\theta_0}{x\sin2\theta_0}\,,
\]
and there $A-B\mathsf{t}^*=\mathsf{c}^2(x^2+r^2)=\mathsf{c}^2>0$ and $C+B\mathsf{t}^*=\mathsf{s}^2>0$. The point $\mathsf{t}^*$ therefore lies in the interior of $J$, and we have $g(\mathsf{t}^*)=\mathsf{c}^2+\mathsf{s}^2=1$. Since $g$ is concave, $\mathsf{t}^*$ is its global maximum on $J$.

Suppose now that $\phi+2\theta_0<\pi/2$. Then $\cos(\phi+2\theta_0)=r\cos2\theta_0-x\sin2\theta_0>0$, so $\mathsf{t}^*<-1$, and the concave function $g$ is decreasing on $[-1,1]$. Therefore, we have $F(y)\le g(-1)$, which is the value of $F$ at $y=-r$:
\[
g(-1)=\mathsf{c}(\mathsf{c}x+\mathsf{s}r)+\mathsf{s}|\mathsf{s}x-\mathsf{c}r| .
\]
Moreover, we have $\mathsf{s}x<\mathsf{c}r$, because $\mathsf{s}x-\mathsf{c}r=-\cos(\phi+\theta_0)$ and $\phi+\theta_0<\pi/2$. Hence, we have
\[
g(-1)=(\mathsf{c}^2-\mathsf{s}^2)x+2\mathsf{c}\mathsf{s}r=\cos2\theta_0\sin\phi+\sin2\theta_0\cos\phi=\sin(\phi+2\theta_0).
\]
Thus $\phi_k\le\phi_{k-1}+2\theta_0$. If instead $\phi+2\theta_0\ge\pi/2$, the bound $\phi_k\le\pi/2=\min\{\phi_{k-1}+2\theta_0,\pi/2\}$ holds by the definition of $\phi_k$. This proves the recursive bound, and the displayed inequality follows by induction from $\phi_0=\theta_0$.
\end{proof}

The bound in Theorem~\ref{thm:grover}(1) is the amplitude of Grover's algorithm after $p$ iterations \cite{Gro97,BBHT98,Zal99}, and it is attained when $d=2$. Indeed, the proof shows that the equality $\phi_k=\phi_{k-1}+2\theta_0$ holds at a layer with $\phi_{k-1}+2\theta_0<\pi/2$ only if $y=-r$ in the notation of the proof. Equivalently, equality at such a layer requires the component of $|\varphi_k\rangle$ orthogonal to $|e\rangle$ to be a negative multiple of $|\xi'\rangle$ after the phase normalization $\langle e|\varphi_k\rangle\ge0$. When $d=2$, the space $e^\perp$ is spanned by $|\xi'\rangle$ and $Z(\gamma_k)$ acts on it by a phase, so the required alignment can be arranged by the choice of $\gamma_k$ at every layer. Grover search with adjustable phases is treated by Long \cite{Lon01} and H\o yer \cite{Hoy00}. For general $d$, invariance of the line $\C|\xi'\rangle$ alone is not sufficient. The relative phase must also give the negative alignment required above. A sufficient condition is
\[
\gamma_k(\lambda_j-\lambda_*)\in\pi+2\pi\Z
\qquad(j\ne j_0),
\]
for every $k$, together with the corresponding choice of mixer phase. Under this condition, the cost layer produces the required relative sign on the carried eigenspaces orthogonal to the target.

\subsection{\texorpdfstring{Proof of Theorem~\ref{thm:intro-depth}  and a spectral refinement}{Proof of Theorem 1.2}}

We now prove the second main theorem.

\begin{proof}[Proof of Theorem~\ref{thm:intro-depth}]
By Theorem~\ref{thm:grover}(1), we have
\[
\sqrt{1-\varepsilon}\le|\langle e|\psi_p\rangle|\le\sin\bigl(\min\{(2p+1)\theta_0,\pi/2\}\bigr).
\]
Hence \(\arcsin\sqrt{1-\varepsilon}\le (2p+1)\theta_0\), which proves the first claim.  For the second claim, the equality $|\langle e|\psi_p\rangle|=1$ forces $(2p+1)\theta_0\ge\pi/2$ by Theorem~\ref{thm:grover}(1), that is, $p\ge\pi/(4\theta_0)-\frac12$. For $\theta\in[0,\pi/2]$, we have $\sin\theta\ge2\theta/\pi$, so $\pi/(2\theta_0)\ge1/\sin\theta_0$, and we have $\sin^2\theta_0=c_{j_0}^2=N_*2^{-n}$ because $|\langle x|\xi\rangle|^2=2^{-n}$ for every basis state $|x\rangle$. Hence, we have $2p+1\ge2^{n/2}N_*^{-1/2}$.
\end{proof}

\begin{theorem}\label{thm:inject}
Assume $\lambda_*=0$ and let $p\ge1$. Then
\[
\langle e|\psi_p\rangle=(\sin\theta_0)\Bigl(\langle\xi|\psi_p\rangle+\sum_{k=1}^p\big\langle(I-Z(\gamma_k)^\dagger)\xi\,\big|\,P_\perp\psi_{k-1}\big\rangle\Bigr),
\]
and consequently, with $\phi_k$ as in Theorem~\ref{thm:grover}(1), we have
\[
|\langle e|\psi_p\rangle|\ \le\ (\sin\theta_0)\Bigl(1+2\cos\theta_0\sum_{k=1}^p\sigma(\gamma_k)\cos\phi_{k-1}\Bigr)\ \le\ (\sin\theta_0)\bigl(1+2p\,\sigma_{\max}\cos\theta_0\bigr).
\]
\end{theorem}

\begin{proof}
Put $\alpha_k:=\langle e|\psi_k\rangle$ for $0\le k\le p$. Since $\lambda_*=0$, we have $\langle e|\varphi_k\rangle=\langle e|\psi_{k-1}\rangle=\alpha_{k-1}$, and the identity $|\psi_k\rangle=|\varphi_k\rangle+w_ko_k|\xi\rangle$ gives $\alpha_k=\alpha_{k-1}+w_ko_k\sin\theta_0$. Taking the inner product of the same identity with $|\xi\rangle$ gives $w_ko_k=\langle\xi|\psi_k\rangle-\langle\xi|\varphi_k\rangle$. Writing 
\[
\langle\xi|\varphi_k\rangle=\langle\xi|\psi_{k-1}\rangle-\langle\xi|(I-Z(\gamma_k))\psi_{k-1}\rangle
\]
and summing over $k$, the differences $\langle\xi|\psi_k\rangle-\langle\xi|\psi_{k-1}\rangle$ telescope to $\langle\xi|\psi_p\rangle-1$. Therefore,
\[
\alpha_p=(\sin\theta_0)\Bigl(1+\langle\xi|\psi_p\rangle-1+\sum_{k=1}^p\langle\xi|(I-Z(\gamma_k))\psi_{k-1}\rangle\Bigr),
\]
and $\langle\xi|(I-Z(\gamma_k))\psi_{k-1}\rangle=\langle(I-Z(\gamma_k)^\dagger)\xi|\psi_{k-1}\rangle$. The vector $(I-Z(\gamma_k)^\dagger)|\xi\rangle$ has coordinates $c_j(1-e^{i\gamma_k\lambda_j})$, whose $j_0$-th coordinate vanishes because $\lambda_*=0$. Thus, the inner product can be taken with $P_\perp|\psi_{k-1}\rangle$, which proves the identity.

For the inequality, we have $|\langle\xi|\psi_p\rangle|\le1$, $\|P_\perp|\psi_{k-1}\rangle\|=\cos\phi_{k-1}$, and
\[
\|(I-Z(\gamma_k)^\dagger)\xi\|^2=\sum_{j\ne j_0}c_j^2|1-e^{i\gamma_k\lambda_j}|^2=4\sum_{j\ne j_0}c_j^2\sin^2\frac{\gamma_k\lambda_j}{2}=4\cos^2\theta_0\,\sigma(\gamma_k)^2 .
\]
The Cauchy--Schwarz inequality applied to each summand gives the first bound, and the bounds $\sigma(\gamma_k)\le\sigma_{\max}$ and $\cos\phi_{k-1}\le1$ give the second.
\end{proof}

\begin{corollary}\label{cor:spectral-depth}
Assume $\lambda_*=0$ and let $\varepsilon\in(0,1)$. If\/
$|\langle e|\psi_p\rangle|^2\ge1-\varepsilon$ for some choice of parameters, then
\[
p\ge
\frac{(\sin\theta_0)^{-1}\sqrt{1-\varepsilon}-1}
{2\sigma_{\max}\cos\theta_0}\,.
\]
\end{corollary}

\begin{proof}
This follows by solving the last inequality in Theorem~\ref{thm:inject} for $p$.
\end{proof}

For $|\chi\rangle=|e\rangle$, Theorem~\ref{thm:inject} replaces the coefficient $1$ of the $2p$ term in Theorem~\ref{thm:grover}(2) by $\sigma_{\max}\cos\theta_0\le1$. Both bounds are linear in $p$. For $d=2$, one has $\sigma(\gamma)=|\sin(\gamma(\lambda_2-\lambda_1)/2)|$ and $\sigma_{\max}=1$.

\subsection{Further refinements}

Next, we give an obstruction to exact reachability arising from consecutive lattice points in the carried spectrum.

\begin{theorem}\label{thm:run}
Assume the lattice condition with spacing $\eta$. Suppose that for some $\lambda_0\in\R$ and some integer $q\ge1$ the values $\lambda_0+\eta,\lambda_0+2\eta,\dots,\lambda_0+ q\eta$ all belong to $\{\lambda_j:j\ne j_0\}$. If $|\langle e|\psi_p\rangle|=1$ for some choice of parameters, then $p\ge q$.
\end{theorem}

\begin{proof}
Let $E$ be the exponential sum \eqref{eq:expsum}, so that $\langle\xi_j|\psi_p\rangle=c_jE(\lambda_j)$ for all $j$. Put $z_k:=e^{-iS_k\eta}$ for $0\le k\le p$. Grouping the indices $k$ with equal $z_k$, we find distinct complex numbers $\zeta_1,\dots,\zeta_t$ of modulus one, with $1\le t\le p+1$, and complex numbers $b_1,\dots,b_t$ such that
\[
E(\lambda_0+\nu\eta)=\sum_{l=1}^t b_l\,\zeta_l^{\,\nu}\qquad\text{for all }\nu\in\Z,
\]
namely, $\zeta_l$ are the different values among $z_0,\dots,z_p$ and $b_l=\sum_{k:\,z_k=\zeta_l}h_ke^{-iS_k\lambda_0}$.

Suppose $|\langle e|\psi_p\rangle|=1$. Then $|\psi_p\rangle$ is a multiple of $|e\rangle$, so $c_jE(\lambda_j)=0$ for all $j\ne j_0$, while $c_{j_0}E(\lambda_*)\ne0$. Since every $c_j$ is positive, the function $E$ vanishes at $\lambda_0+\nu\eta$ for $\nu=1,\dots, q$ and $E(\lambda_*)\ne0$. By the lattice condition, we have $\lambda_*-(\lambda_0+\eta)\in\eta\Z$, so $\lambda_*\in\lambda_0+\eta\Z$ and not all $b_l$ vanish. 

If $t\le q$, then the $t$ equations $\sum_lb_l\zeta_l^{\,\nu}=0$ for $\nu=1,\dots,t$ form a linear system for $(b_l\zeta_l)_l$ whose matrix $(\zeta_l^{\,\nu-1})_{\nu,l}$ is a Vandermonde matrix with distinct nodes, hence invertible \cite{HJ13}. This forces $b_l\zeta_l=0$ for all $l$, hence $b_l=0$ for all $l$, a contradiction. Therefore, we have $t\ge q+1$, and so $p\ge t-1\ge q$.
\end{proof}

When the spectrum carried by $|\xi\rangle$ is ${\lambda_1,\lambda_1+\eta,\dots,\lambda_1+(d-1)\eta}$ and the target is $|\xi_1\rangle$ or $|\xi_d\rangle$, Theorem~\ref{thm:run} applies with $q=d-1$ and shows that exact reachability requires $p\ge d-1$ for arbitrary choices of the angles.  Note that this bound is sharp. For example, for $\lambda=(0,1,3)$ and $c_1^2=\frac14$, the operator $Z(\pi)$ acts as $+1$ on $|\xi_1\rangle$ and as $-1$ on $|\xi_2\rangle,|\xi_3\rangle$. The dynamics then reduces to the case $d=2$ with $\sin\theta_0=\frac12$, and the equality $(2\cdot1+1)\theta_0=\pi/2$ shows that $|\xi_1\rangle$ is reached exactly at $p=1<d-1$.

\subsection{Numerical observations}\label{subsec:numerics}

For the following numerical observations, given a pair $(\lambda,c)$ and depth $p$,
we numerically optimized the infidelity $1-|\langle e|\psi_p\rangle|^2$ over all $2p$ angles
by quasi-Newton (L-BFGS) optimization with random and warm restarts. We recorded the smallest depth $p^*$ at which the smallest infidelity found was below $10^{-9}$.
{We first describe the computation and then record the observations. Figures~\ref{fig:envelope}--\ref{fig:pstar} display the results.

By Lemma~\ref{lem:invariant}, all computations can be carried out in the invariant subspace \(W_0\), which we identify with \(\C^d\) using the basis \(\{|\xi_j\rangle\}\). In the basis $\{|\xi_j\rangle\}$, the operator $Z(\gamma)$ is the diagonal map with entries $e^{-i\gamma\lambda_j}$, and $R(\beta)$ acts as the rank-one update $|\varphi\rangle\mapsto|\varphi\rangle+(e^{-i\beta}-1)\langle\xi|\varphi\rangle|\xi\rangle$. Hence, one evaluation of the state $|\psi_p\rangle$, and of the infidelity, costs $O(pd)$ arithmetic operations, and the infidelity is a smooth function of the $2p$ angles. We minimized it with the limited-memory quasi-Newton method L-BFGS-B \cite{BLNZ95}, as implemented in SciPy \cite{SciPy20}, with the gradient with respect to all $2p$ angles computed by automatic differentiation \cite{JAX18}. All computations were performed in IEEE double-precision arithmetic \cite{IEEE754}, with objective tolerance $10^{-15}$ and gradient tolerance $10^{-12}$.

Each minimization used between $20$ and $60$ starting points of three kinds. The starting points of the first kind draw all $2p$ angles independently and uniformly from $[0,2\pi)$. Starting points of the second kind extend the best parameters found at depth $p-1$ by one identity layer, with $\beta=0$ and a uniformly random $\gamma$, inserted at a uniformly random position, and then perturb all angles by small Gaussian noise. For the largest instances (the consecutive spectra with $d\in\{8,16\}$ and small $\sin^2\theta_0$, and the binomial weights with $m\in\{8,10\}$), starting points of a third kind were also used. These starting points assign one fixed pair $(\beta,\gamma)$ to all interior layers, which gives a constant-angle schedule, and randomize up to three initial and final layers. The smallest infidelity found over all starting points is reported at each depth.

The computation certifies reachability, and the computed angles are the certificate. If the search reaches an infidelity below $10^{-9}$ at some depth, then $|e\rangle$ is reachable to that accuracy at that depth. The computation does not certify unreachability, since at a depth where the search stalls the true minimum can be smaller than the value found. The recorded $p^*$ is therefore an upper bound for the smallest depth at which the infidelity falls below $10^{-9}$. Theorem~\ref{thm:intro-depth} bounds this depth from below, and Theorem~\ref{thm:run} bounds from below the depth required for exact reachability. At every reported $p^*$, the smallest infidelity found was below $10^{-14}$, which is at the level of the double-precision rounding error.}

{We begin with two families of instances: the case $d=2$ with $\sin^2\theta_0\in\{0.3,0.1,0.03,0.01\}$ and target $j_0=1$, and the spectrum $\lambda=(0,1,3,5)$ with $\sin^2\theta_0\in\{0.1,0.03\}$, target $j_0=1$, and the remaining three weights equal. For these instances, the smallest infidelity found was in agreement with
\[
1-\sin^2\bigl(\min\{(2p+1)\theta_0,\pi/2\}\bigr)
\]
to $10^{-15}$ at every depth up to the first depth at which this quantity vanishes. {Figure~\ref{fig:envelope} shows two of these instances. The optimized fidelities lie on the Grover envelope at every depth.} For the spectra $\lambda=(0,1,\dots,d-1)$ with target $j_0=1$, equal weights $c_j^2=1/d$ and $3\le d\le8$, we found $p^*=d-1$ in every case, so the bound of Theorem~\ref{thm:run} was numerically attained to the stated tolerance. {The left panel of Figure~\ref{fig:infid} records the minimized infidelity at each depth for these six spectra.}

\begin{figure}[t]

\centering
\begin{tikzpicture}
\begin{axis}[gmaxis, xlabel={$p$}, ylabel={$|\langle e|\psi_p\rangle|^2$},
  xmin=-0.2, xmax=8.4, ymin=0, ymax=1.06, xtick={0,1,2,3,4,5,6,7,8},
  title={$d=2$, $\sin^2\theta_0=0.01$}, legend pos=south east]
\addplot[gmblue, line width=0.9pt] coordinates {(0.0000,0.010000) (0.0278,0.011137) (0.0555,0.012334) (0.0833,0.013592) (0.1110,0.014910) (0.1388,0.016288) (0.1666,0.017726) (0.1943,0.019223) (0.2221,0.020780) (0.2498,0.022396) (0.2776,0.024071) (0.3054,0.025805) (0.3331,0.027598) (0.3609,0.029449) (0.3886,0.031359) (0.4164,0.033326) (0.4441,0.035351) (0.4719,0.037434) (0.4997,0.039574) (0.5274,0.041771) (0.5552,0.044024) (0.5829,0.046334) (0.6107,0.048700) (0.6385,0.051122) (0.6662,0.053599) (0.6940,0.056132) (0.7217,0.058719) (0.7495,0.061361) (0.7773,0.064058) (0.8050,0.066808) (0.8328,0.069612) (0.8605,0.072469) (0.8883,0.075379) (0.9161,0.078341) (0.9438,0.081356) (0.9716,0.084422) (0.9993,0.087540) (1.0271,0.090709) (1.0548,0.093929) (1.0826,0.097198) (1.1104,0.100518) (1.1381,0.103887) (1.1659,0.107305) (1.1936,0.110771) (1.2214,0.114286) (1.2492,0.117849) (1.2769,0.121458) (1.3047,0.125115) (1.3324,0.128818) (1.3602,0.132567) (1.3880,0.136361) (1.4157,0.140200) (1.4435,0.144084) (1.4712,0.148012) (1.4990,0.151983) (1.5268,0.155997) (1.5545,0.160054) (1.5823,0.164153) (1.6100,0.168294) (1.6378,0.172475) (1.6656,0.176698) (1.6933,0.180960) (1.7211,0.185261) (1.7488,0.189602) (1.7766,0.193981) (1.8043,0.198397) (1.8321,0.202851) (1.8599,0.207342) (1.8876,0.211869) (1.9154,0.216432) (1.9431,0.221029) (1.9709,0.225662) (1.9987,0.230328) (2.0264,0.235027) (2.0542,0.239760) (2.0819,0.244524) (2.1097,0.249320) (2.1375,0.254147) (2.1652,0.259005) (2.1930,0.263892) (2.2207,0.268809) (2.2485,0.273754) (2.2763,0.278727) (2.3040,0.283728) (2.3318,0.288755) (2.3595,0.293808) (2.3873,0.298887) (2.4151,0.303991) (2.4428,0.309119) (2.4706,0.314271) (2.4983,0.319445) (2.5261,0.324642) (2.5538,0.329861) (2.5816,0.335101) (2.6094,0.340361) (2.6371,0.345641) (2.6649,0.350939) (2.6926,0.356257) (2.7204,0.361592) (2.7482,0.366944) (2.7759,0.372313) (2.8037,0.377697) (2.8314,0.383097) (2.8592,0.388511) (2.8870,0.393939) (2.9147,0.399380) (2.9425,0.404834) (2.9702,0.410299) (2.9980,0.415775) (3.0258,0.421262) (3.0535,0.426759) (3.0813,0.432264) (3.1090,0.437778) (3.1368,0.443300) (3.1645,0.448828) (3.1923,0.454363) (3.2201,0.459904) (3.2478,0.465450) (3.2756,0.470999) (3.3033,0.476553) (3.3311,0.482109) (3.3589,0.487668) (3.3866,0.493228) (3.4144,0.498789) (3.4421,0.504350) (3.4699,0.509910) (3.4977,0.515470) (3.5254,0.521027) (3.5532,0.526582) (3.5809,0.532133) (3.6087,0.537681) (3.6365,0.543224) (3.6642,0.548761) (3.6920,0.554293) (3.7197,0.559818) (3.7475,0.565335) (3.7753,0.570844) (3.8030,0.576345) (3.8308,0.581836) (3.8585,0.587317) (3.8863,0.592787) (3.9140,0.598246) (3.9418,0.603692) (3.9696,0.609126) (3.9973,0.614546) (4.0251,0.619952) (4.0528,0.625343) (4.0806,0.630719) (4.1084,0.636079) (4.1361,0.641421) (4.1639,0.646747) (4.1916,0.652054) (4.2194,0.657342) (4.2472,0.662611) (4.2749,0.667859) (4.3027,0.673087) (4.3304,0.678294) (4.3582,0.683478) (4.3860,0.688640) (4.4137,0.693778) (4.4415,0.698893) (4.4692,0.703982) (4.4970,0.709047) (4.5247,0.714086) (4.5525,0.719098) (4.5803,0.724083) (4.6080,0.729040) (4.6358,0.733970) (4.6635,0.738870) (4.6913,0.743740) (4.7191,0.748581) (4.7468,0.753390) (4.7746,0.758169) (4.8023,0.762915) (4.8301,0.767629) (4.8579,0.772310) (4.8856,0.776957) (4.9134,0.781570) (4.9411,0.786148) (4.9689,0.790690) (4.9967,0.795197) (5.0244,0.799667) (5.0522,0.804100) (5.0799,0.808496) (5.1077,0.812853) (5.1355,0.817171) (5.1632,0.821451) (5.1910,0.825690) (5.2187,0.829890) (5.2465,0.834048) (5.2742,0.838165) (5.3020,0.842241) (5.3298,0.846274) (5.3575,0.850264) (5.3853,0.854211) (5.4130,0.858114) (5.4408,0.861972) (5.4686,0.865786) (5.4963,0.869555) (5.5241,0.873278) (5.5518,0.876955) (5.5796,0.880585) (5.6074,0.884168) (5.6351,0.887703) (5.6629,0.891191) (5.6906,0.894630) (5.7184,0.898021) (5.7462,0.901362) (5.7739,0.904653) (5.8017,0.907895) (5.8294,0.911086) (5.8572,0.914226) (5.8849,0.917315) (5.9127,0.920352) (5.9405,0.923337) (5.9682,0.926270) (5.9960,0.929150) (6.0237,0.931977) (6.0515,0.934751) (6.0793,0.937471) (6.1070,0.940137) (6.1348,0.942748) (6.1625,0.945304) (6.1903,0.947806) (6.2181,0.950252) (6.2458,0.952642) (6.2736,0.954977) (6.3013,0.957255) (6.3291,0.959476) (6.3569,0.961641) (6.3846,0.963748) (6.4124,0.965799) (6.4401,0.967791) (6.4679,0.969726) (6.4957,0.971603) (6.5234,0.973421) (6.5512,0.975180) (6.5789,0.976881) (6.6067,0.978523) (6.6344,0.980106) (6.6622,0.981629) (6.6900,0.983093) (6.7177,0.984497) (6.7455,0.985841) (6.7732,0.987125) (6.8010,0.988349) (6.8288,0.989512) (6.8565,0.990615) (6.8843,0.991657) (6.9120,0.992638) (6.9398,0.993559) (6.9676,0.994418) (6.9953,0.995216) (7.0231,0.995953) (7.0508,0.996628) (7.0786,0.997242) (7.1064,0.997795) (7.1341,0.998286) (7.1619,0.998715) (7.1896,0.999082) (7.2174,0.999388) (7.2452,0.999632) (7.2729,0.999815) (7.3007,0.999935) (7.3284,0.999994) (7.3562,1.000000) (7.3839,1.000000) (7.4117,1.000000) (7.4395,1.000000) (7.4672,1.000000) (7.4950,1.000000) (7.5227,1.000000) (7.5505,1.000000) (7.5783,1.000000) (7.6060,1.000000) (7.6338,1.000000) (7.6615,1.000000) (7.6893,1.000000) (7.7171,1.000000) (7.7448,1.000000) (7.7726,1.000000) (7.8003,1.000000) (7.8281,1.000000) (7.8559,1.000000) (7.8836,1.000000) (7.9114,1.000000) (7.9391,1.000000) (7.9669,1.000000) (7.9946,1.000000) (8.0224,1.000000) (8.0502,1.000000) (8.0779,1.000000) (8.1057,1.000000) (8.1334,1.000000) (8.1612,1.000000) (8.1890,1.000000) (8.2167,1.000000) (8.2445,1.000000) (8.2722,1.000000) (8.3000,1.000000)};
\addplot[only marks, mark=o, mark size=2.0pt, black] coordinates {(1,0.087616000000) (2,0.230553625600) (3,0.416171556905) (4,0.615067913596) (5,0.795737512774) (6,0.929562289928) (7,0.995344400358) (8,1.000000000000)};
\legend{Grover envelope, optimized fidelity}
\end{axis}
\end{tikzpicture}\hfill
\begin{tikzpicture}
\begin{axis}[gmaxis, xlabel={$p$},
  xmin=-0.2, xmax=5.3, ymin=0, ymax=1.06, xtick={0,1,2,3,4,5},
  title={$\lambda=(0,1,3,5)$, $\sin^2\theta_0=0.03$}]
\addplot[gmblue, line width=0.9pt] coordinates {(0.0000,0.030000) (0.0177,0.032141) (0.0355,0.034354) (0.0532,0.036637) (0.0709,0.038992) (0.0886,0.041416) (0.1064,0.043910) (0.1241,0.046474) (0.1418,0.049107) (0.1595,0.051808) (0.1773,0.054578) (0.1950,0.057415) (0.2127,0.060321) (0.2304,0.063293) (0.2482,0.066331) (0.2659,0.069436) (0.2836,0.072606) (0.3013,0.075841) (0.3191,0.079141) (0.3368,0.082505) (0.3545,0.085933) (0.3722,0.089424) (0.3900,0.092977) (0.4077,0.096593) (0.4254,0.100269) (0.4431,0.104007) (0.4609,0.107805) (0.4786,0.111663) (0.4963,0.115580) (0.5140,0.119555) (0.5318,0.123589) (0.5495,0.127679) (0.5672,0.131827) (0.5849,0.136031) (0.6027,0.140290) (0.6204,0.144603) (0.6381,0.148971) (0.6559,0.153393) (0.6736,0.157867) (0.6913,0.162394) (0.7090,0.166972) (0.7268,0.171600) (0.7445,0.176279) (0.7622,0.181007) (0.7799,0.185783) (0.7977,0.190608) (0.8154,0.195479) (0.8331,0.200397) (0.8508,0.205361) (0.8686,0.210369) (0.8863,0.215422) (0.9040,0.220518) (0.9217,0.225656) (0.9395,0.230836) (0.9572,0.236058) (0.9749,0.241319) (0.9926,0.246620) (1.0104,0.251960) (1.0281,0.257337) (1.0458,0.262751) (1.0635,0.268202) (1.0813,0.273688) (1.0990,0.279208) (1.1167,0.284762) (1.1344,0.290348) (1.1522,0.295967) (1.1699,0.301617) (1.1876,0.307297) (1.2054,0.313006) (1.2231,0.318744) (1.2408,0.324509) (1.2585,0.330301) (1.2763,0.336119) (1.2940,0.341962) (1.3117,0.347829) (1.3294,0.353719) (1.3472,0.359632) (1.3649,0.365566) (1.3826,0.371520) (1.4003,0.377494) (1.4181,0.383486) (1.4358,0.389497) (1.4535,0.395524) (1.4712,0.401567) (1.4890,0.407625) (1.5067,0.413697) (1.5244,0.419782) (1.5421,0.425880) (1.5599,0.431989) (1.5776,0.438108) (1.5953,0.444237) (1.6130,0.450374) (1.6308,0.456518) (1.6485,0.462670) (1.6662,0.468827) (1.6839,0.474988) (1.7017,0.481154) (1.7194,0.487322) (1.7371,0.493492) (1.7548,0.499664) (1.7726,0.505835) (1.7903,0.512006) (1.8080,0.518174) (1.8258,0.524340) (1.8435,0.530502) (1.8612,0.536660) (1.8789,0.542812) (1.8967,0.548957) (1.9144,0.555095) (1.9321,0.561225) (1.9498,0.567345) (1.9676,0.573455) (1.9853,0.579554) (2.0030,0.585641) (2.0207,0.591714) (2.0385,0.597774) (2.0562,0.603819) (2.0739,0.609847) (2.0916,0.615860) (2.1094,0.621854) (2.1271,0.627830) (2.1448,0.633787) (2.1625,0.639723) (2.1803,0.645638) (2.1980,0.651530) (2.2157,0.657400) (2.2334,0.663246) (2.2512,0.669066) (2.2689,0.674861) (2.2866,0.680630) (2.3043,0.686370) (2.3221,0.692083) (2.3398,0.697766) (2.3575,0.703419) (2.3753,0.709041) (2.3930,0.714631) (2.4107,0.720189) (2.4284,0.725713) (2.4462,0.731202) (2.4639,0.736657) (2.4816,0.742075) (2.4993,0.747456) (2.5171,0.752800) (2.5348,0.758105) (2.5525,0.763371) (2.5702,0.768597) (2.5880,0.773782) (2.6057,0.778925) (2.6234,0.784025) (2.6411,0.789083) (2.6589,0.794096) (2.6766,0.799064) (2.6943,0.803987) (2.7120,0.808864) (2.7298,0.813694) (2.7475,0.818475) (2.7652,0.823209) (2.7829,0.827893) (2.8007,0.832527) (2.8184,0.837110) (2.8361,0.841642) (2.8538,0.846122) (2.8716,0.850549) (2.8893,0.854923) (2.9070,0.859243) (2.9247,0.863508) (2.9425,0.867718) (2.9602,0.871871) (2.9779,0.875968) (2.9957,0.880008) (3.0134,0.883990) (3.0311,0.887913) (3.0488,0.891778) (3.0666,0.895582) (3.0843,0.899326) (3.1020,0.903010) (3.1197,0.906632) (3.1375,0.910192) (3.1552,0.913690) (3.1729,0.917124) (3.1906,0.920495) (3.2084,0.923802) (3.2261,0.927045) (3.2438,0.930222) (3.2615,0.933334) (3.2793,0.936380) (3.2970,0.939359) (3.3147,0.942271) (3.3324,0.945116) (3.3502,0.947893) (3.3679,0.950602) (3.3856,0.953243) (3.4033,0.955814) (3.4211,0.958316) (3.4388,0.960748) (3.4565,0.963110) (3.4742,0.965401) (3.4920,0.967621) (3.5097,0.969770) (3.5274,0.971848) (3.5452,0.973853) (3.5629,0.975787) (3.5806,0.977648) (3.5983,0.979436) (3.6161,0.981151) (3.6338,0.982793) (3.6515,0.984361) (3.6692,0.985856) (3.6870,0.987276) (3.7047,0.988622) (3.7224,0.989894) (3.7401,0.991091) (3.7579,0.992214) (3.7756,0.993261) (3.7933,0.994233) (3.8110,0.995130) (3.8288,0.995952) (3.8465,0.996698) (3.8642,0.997368) (3.8819,0.997963) (3.8997,0.998481) (3.9174,0.998924) (3.9351,0.999291) (3.9528,0.999581) (3.9706,0.999796) (3.9883,0.999934) (4.0060,0.999996) (4.0237,1.000000) (4.0415,1.000000) (4.0592,1.000000) (4.0769,1.000000) (4.0946,1.000000) (4.1124,1.000000) (4.1301,1.000000) (4.1478,1.000000) (4.1656,1.000000) (4.1833,1.000000) (4.2010,1.000000) (4.2187,1.000000) (4.2365,1.000000) (4.2542,1.000000) (4.2719,1.000000) (4.2896,1.000000) (4.3074,1.000000) (4.3251,1.000000) (4.3428,1.000000) (4.3605,1.000000) (4.3783,1.000000) (4.3960,1.000000) (4.4137,1.000000) (4.4314,1.000000) (4.4492,1.000000) (4.4669,1.000000) (4.4846,1.000000) (4.5023,1.000000) (4.5201,1.000000) (4.5378,1.000000) (4.5555,1.000000) (4.5732,1.000000) (4.5910,1.000000) (4.6087,1.000000) (4.6264,1.000000) (4.6441,1.000000) (4.6619,1.000000) (4.6796,1.000000) (4.6973,1.000000) (4.7151,1.000000) (4.7328,1.000000) (4.7505,1.000000) (4.7682,1.000000) (4.7860,1.000000) (4.8037,1.000000) (4.8214,1.000000) (4.8391,1.000000) (4.8569,1.000000) (4.8746,1.000000) (4.8923,1.000000) (4.9100,1.000000) (4.9278,1.000000) (4.9455,1.000000) (4.9632,1.000000) (4.9809,1.000000) (4.9987,1.000000) (5.0164,1.000000) (5.0341,1.000000) (5.0518,1.000000) (5.0696,1.000000) (5.0873,1.000000) (5.1050,1.000000) (5.1227,1.000000) (5.1405,1.000000) (5.1582,1.000000) (5.1759,1.000000) (5.1936,1.000000) (5.2114,1.000000) (5.2291,1.000000) (5.2468,1.000000) (5.2645,1.000000) (5.2823,1.000000) (5.3000,1.000000)};
\addplot[only marks, mark=o, mark size=2.0pt, black] coordinates {(1,0.248832000000) (2,0.584607820800) (3,0.880990240236) (4,0.999983603817) (5,1.000000000000)};
\end{axis}
\end{tikzpicture}
\caption{Attainment of the Grover envelope. Circles show the largest value of $|\langle e|\psi_p\rangle|^2$ found at each depth $p$. The curve is the envelope $\sin^2\bigl(\min\{(2p+1)\theta_0,\pi/2\}\bigr)$ of Theorem~\ref{thm:grover}(1), drawn as a function of a continuous depth variable. Left: $d=2$ with $\sin^2\theta_0=0.01$. Right: the spectrum $\lambda=(0,1,3,5)$ with $\sin^2\theta_0=0.03$ and the remaining three weights equal. In both panels, the circles agree with the envelope at every integer depth to the double-precision rounding error.}
\label{fig:envelope}
\end{figure}
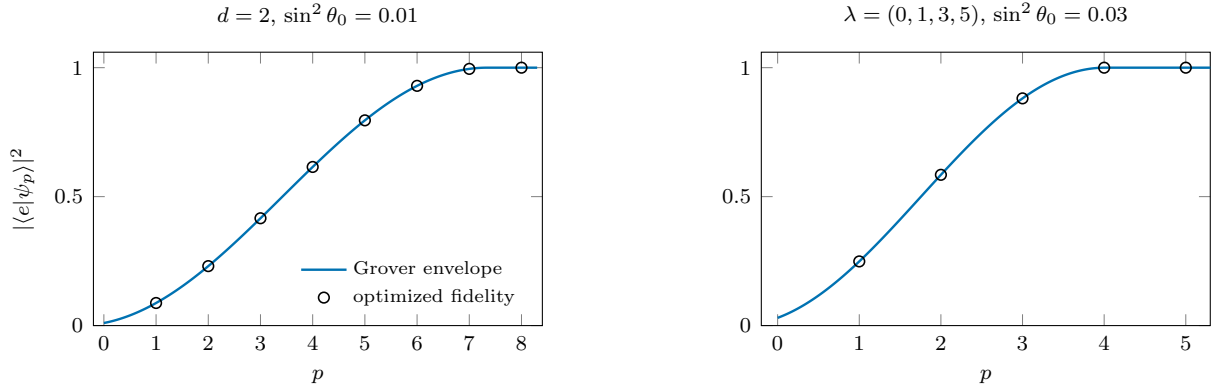

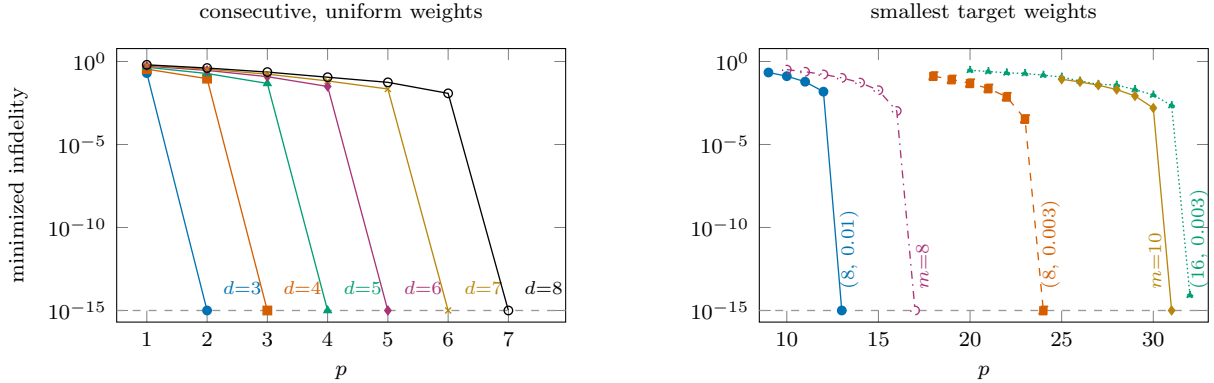
\begin{figure}[t]

\centering
\begin{tikzpicture}
\begin{axis}[gmaxis, xlabel={$p$}, ylabel={minimized infidelity}, ymode=log,
  xmin=0.5, xmax=7.95, ymin=2e-16, ymax=6, xtick={1,2,3,4,5,6,7},
  ytick={1e0,1e-5,1e-10,1e-15},
  title={consecutive, uniform weights}]
\addplot[black!40, dashed, line width=0.5pt, forget plot] coordinates {(0.5,1e-15) (7.95,1e-15)};
\addplot[gmblue, mark=*, mark size=1.6pt, line width=0.5pt] coordinates {(1,2.001603e-01) (2,1.000000e-15)};
\node[font=\tiny, gmblue, anchor=south west] at (axis cs:2.10,2.5e-15) {$d{=}3$};
\addplot[gmverm, mark=square*, mark size=1.6pt, line width=0.5pt] coordinates {(1,3.462939e-01) (2,9.232647e-02) (3,1.000000e-15)};
\node[font=\tiny, gmverm, anchor=south west] at (axis cs:3.10,2.5e-15) {$d{=}4$};
\addplot[gmgreen, mark=triangle*, mark size=1.6pt, line width=0.5pt] coordinates {(1,4.500879e-01) (2,1.887771e-01) (3,4.769239e-02) (4,1.000000e-15)};
\node[font=\tiny, gmgreen, anchor=south west] at (axis cs:4.10,2.5e-15) {$d{=}5$};
\addplot[gmpurp, mark=diamond*, mark size=1.6pt, line width=0.5pt] coordinates {(1,5.263577e-01) (2,2.909950e-01) (3,1.220803e-01) (4,3.130331e-02) (5,1.000000e-15)};
\node[font=\tiny, gmpurp, anchor=south west] at (axis cs:5.10,2.5e-15) {$d{=}6$};
\addplot[gmgold, mark=x, mark size=1.6pt, line width=0.5pt] coordinates {(1,5.844226e-01) (2,3.443091e-01) (3,1.721894e-01) (4,6.735113e-02) (5,2.206138e-02) (6,1.000000e-15)};
\node[font=\tiny, gmgold, anchor=south west] at (axis cs:6.10,2.5e-15) {$d{=}7$};
\addplot[black, mark=o, mark size=1.6pt, line width=0.5pt] coordinates {(1,6.299817e-01) (2,4.040297e-01) (3,2.304526e-01) (4,1.127721e-01) (5,5.544985e-02) (6,1.203431e-02) (7,1.000000e-15)};
\node[font=\tiny, black, anchor=south west] at (axis cs:7.10,2.5e-15) {$d{=}8$};
\end{axis}
\end{tikzpicture}\hfill
\begin{tikzpicture}
\begin{axis}[gmaxis, xlabel={$p$}, ymode=log,
  xmin=8.5, xmax=33, ymin=2e-16, ymax=6, xtick={10,15,20,25,30},
  ytick={1e0,1e-5,1e-10,1e-15},
  title={smallest target weights}]
\addplot[black!40, dashed, line width=0.5pt, forget plot] coordinates {(8.5,1e-15) (33,1e-15)};
\addplot[gmblue, solid, mark=*, mark size=1.6pt, line width=0.5pt] coordinates {(9,2.1740e-01) (10,1.3020e-01) (11,6.0320e-02) (12,1.5180e-02) (13,1.0000e-15)};
\addplot[gmverm, dashed, mark=square*, mark size=1.6pt, line width=0.5pt] coordinates {(18,1.3140e-01) (19,8.2130e-02) (20,4.8670e-02) (21,2.3640e-02) (22,7.4470e-03) (23,3.4010e-04) (24,1.0000e-15)};
\addplot[gmgreen, densely dotted, mark=triangle*, mark size=1.6pt, line width=0.5pt] coordinates {(20,2.9440e-01) (21,2.4930e-01) (22,2.0620e-01) (23,1.8970e-01) (24,1.5090e-01) (25,1.1520e-01) (26,6.6430e-02) (27,4.2500e-02) (28,3.8030e-02) (29,2.0400e-02) (30,9.6040e-03) (31,2.2370e-03) (32,8.2160e-15)};
\addplot[gmpurp, dashdotted, mark=o, mark size=1.6pt, line width=0.5pt] coordinates {(10,3.2420e-01) (11,2.3990e-01) (12,1.6570e-01) (13,1.0300e-01) (14,5.3430e-02) (15,1.8630e-02) (16,1.0470e-03) (17,1.0000e-15)};
\addplot[gmgold, solid, mark=diamond*, mark size=1.6pt, line width=0.5pt] coordinates {(25,8.4920e-02) (26,5.8780e-02) (27,3.7210e-02) (28,2.0360e-02) (29,8.4160e-03) (30,1.5780e-03) (31,1.0000e-15)};
\node[font=\tiny, gmblue, rotate=90, anchor=west] at (axis cs:13.4,5e-15) {$(8,\,0.01)$};
\node[font=\tiny, gmpurp, rotate=90, anchor=west] at (axis cs:17.4,5e-15) {$m{=}8$};
\node[font=\tiny, gmverm, rotate=90, anchor=west] at (axis cs:24.4,5e-15) {$(8,\,0.003)$};
\node[font=\tiny, gmgold, rotate=90, anchor=west] at (axis cs:30.1,5e-15) {$m{=}10$};
\node[font=\tiny, gmgreen, rotate=90, anchor=west] at (axis cs:32.5,5e-15) {$(16,\,0.003)$};
\end{axis}
\end{tikzpicture}
\caption{The smallest infidelity $1-|\langle e|\psi_p\rangle|^2$ found at each depth $p$, on a logarithmic scale. Values at or below $10^{-15}$ are plotted on the dashed line at $10^{-15}$. Left: the consecutive spectra $\lambda=(0,1,\dots,d-1)$ with uniform weights and $3\le d\le8$. Each curve reaches the rounding floor at $p^*=d-1$, the smallest depth permitted by Theorem~\ref{thm:run} for exact reachability. Right: the consecutive instances with $d=8$, $\sin^2\theta_0\in\{0.01,0.003\}$, and $d=16$, $\sin^2\theta_0=0.003$, and the binomial instances with $m\in\{8,10\}$. The curves are labeled by the pair $(d,\sin^2\theta_0)$ for the consecutive instances and by $m$ for the binomial instances.}
\label{fig:infid}
\end{figure}

Next, we consider the same consecutive spectra with $\sin^2\theta_0\in\{0.1,0.03,0.01,0.003\}$ and the remaining weights equal, and the binomial weights $c_j^2=\binom{m}{j-1}2^{-m}$ in $\lambda=(0,1,\dots,m)$ with target $j_0=1$ and $m\in\{4,6,8,10\}$. For these instances, the measured value of $p^*$ exceeded the lower bound in Theorem~\ref{thm:intro-depth} by a factor between $1.1$ and $2.5$. {The consecutive instances have $d\in\{3,4,5,6,8\}$, together with one larger instance with $d=16$ at $\sin^2\theta_0=0.003$, and the pairs $(d,\sin^2\theta_0)$ that were run are those shown in Figure~\ref{fig:pstar}. The right panel of Figure~\ref{fig:infid} records the minimized infidelity at each depth for instances with the smallest target weights. Figure~\ref{fig:pstar} compares every measured $p^*$ of this paragraph with the integer lower bound obtained from Theorem~\ref{thm:intro-depth}.} For example, we measured $p^*=31$ for $m=10$, while Theorem~\ref{thm:intro-depth} gives the integer lower bound $25$. These computations are consistent with the proved lower bounds, but do not determine the optimal depth in general.}

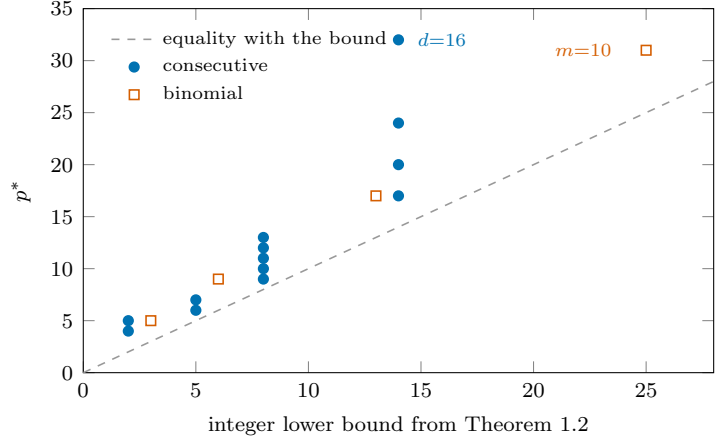
\begin{figure}[t]

\centering
\begin{tikzpicture}
\begin{axis}[gmaxis, width=0.62\textwidth, height=6.4cm,
  xlabel={integer lower bound from Theorem~\ref{thm:intro-depth}}, ylabel={$p^*$},
  xmin=0, xmax=28, ymin=0, ymax=35, xtick={0,5,10,15,20,25}, ytick={0,5,10,15,20,25,30,35},
  legend pos=north west]
\addplot[black!40, dashed, domain=0:28] {x};
\addplot[only marks, mark=*, mark size=1.8pt, gmblue] coordinates {(2,4) (2,5) (5,6) (5,7) (8,9) (8,10) (8,11) (8,12) (8,13) (14,17) (14,20) (14,24) (14,32)};
\addplot[only marks, mark=square, mark size=1.8pt, gmverm] coordinates {(3,5) (6,9) (13,17) (25,31)};
\legend{equality with the bound, consecutive, binomial}
\node[font=\tiny, gmblue, anchor=west] at (axis cs:14.4,32) {$d{=}16$};
\node[font=\tiny, gmverm, anchor=west] at (axis cs:20.5,31) {$m{=}10$};
\end{axis}
\end{tikzpicture}
\caption{The measured $p^*$ against the integer lower bound $\bigl\lceil\arcsin\sqrt{1-\varepsilon}/(2\theta_0)-\tfrac12\bigr\rceil$ from Theorem~\ref{thm:intro-depth} at $\varepsilon=10^{-9}$. For every instance shown, this bound equals $\lceil\pi/(4\theta_0)-\tfrac12\rceil$. Filled circles: the thirteen consecutive instances with $d\in\{3,4,5,6,8,16\}$ and the values of $\sin^2\theta_0$ named in the text. Open squares: the binomial instances with $m\in\{4,6,8,10\}$. The dashed line marks equality with the bound. Every point lies above the line.}
\label{fig:pstar}
\end{figure}

\section{Conclusions}\label{sec:conclusions}

In this paper, we obtained two complementary finite-depth results for QAOA with a Grover mixer: on the variance of the loss function and on the reachability of a carried eigenstate such as the ground state.

First, Theorem~\ref{thm:intro-variance} gives an explicit lower bound on the variance of the loss that holds at every depth under the lattice condition and the independent uniform parameter law. Corollary~\ref{cor:gradient} gives the same bound for the derivative with respect to the final mixing angle. These estimates are obtained directly from phase averaging and the recursion for the mean populations, without using the dynamical Lie algebra or a large-depth approximation. The bounds are instance-dependent, and Remark~\ref{rem:family-scaling} identifies the additional scaling condition needed to obtain an inverse-polynomial statement for a growing family. Theorem~\ref{thm:local} verifies this condition for integer-valued $s$-local objective functions with the uniform-superposition initial state when $s$ and the bound $K$ on the local terms are fixed. In that case, the theorem gives, at every depth, a lower bound of order $n^{-7s}$ for the variance of the normalized loss and of the final-mixer gradient. This bound is a finite-depth counterpart of the large-depth bound of order $n^{-2s}$ in Theorem~IV.4 of \cite{TNB25}. For MaxCut with the uniform-superposition initial state, Corollary~\ref{cor:maxcut} gives sharper inverse-polynomial lower bounds for both the loss variance and the final-mixer gradient variance at every depth.

Cerezo et al.~\cite{CLG25} point out that mechanisms that provably avoid barren plateaus often confine the dynamics to a subspace in which the loss can be simulated classically. In the present setting, the dynamics is confined to $W_0$ by Lemma~\ref{lem:invariant}, and the loss is a function of the carried cost values $\lambda_j$ and their weights $c_j$, so it can be computed classically once these $2d$ numbers are known. However, determining the carried cost values and their weights is in general equivalent to solving the underlying optimization problem; see Section~V.B of \cite{TNB25}.

Second, we obtained results on reachability. Theorem~\ref{thm:grover} shows that, after $p$ layers of GM-QAOA, the probability assigned to any set of carried cost values is at most $(2p+1)^2$ times its probability under the initial distribution. Theorem~\ref{thm:intro-depth} therefore implies that a carried eigenspace of initial probability $q$ cannot be approximated with constant fidelity unless $p$ is of order at least $q^{-1/2}$. For $|\xi\rangle=|+\rangle^{\otimes n}$ and a target eigenvalue attained on $N_*$ computational basis states, this gives a lower bound of order $2^{n/2}N_*^{-1/2}$. Theorem~\ref{thm:inject} presents a spectrum-dependent refinement of the linear overlap bound, while Theorem~\ref{thm:run} provides an independent obstruction to exact reachability when the carried spectrum contains consecutive lattice points. The numerical observations in Subsection~\ref{subsec:numerics} indicate that these obstructions are sharp or close to sharp in several small families. 

We note that the variance and reachability bounds concern different properties: the loss and final-mixer gradient may have inverse-polynomial variance even when concentrating the state on a rare optimum requires exponential depth.

Several questions remain open. It would be useful to determine whether the depth-independent variance bounds persist without the lattice condition, for nonuniform or correlated parameter distributions, and for gradient coordinates other than the final mixing angle. On the reachability side, matching upper bounds for the depth in structured families such as MaxCut, Boolean satisfiability \cite{ZPSQKDPH24}, and constrained optimization problems \cite{BE20,HWORVB19} would clarify how much of the Grover-type obstruction is intrinsic. A further direction is to combine the finite-depth estimates here with the dynamical-Lie-algebra description of \cite{TNB25} to quantitatively understand the transition from shallow circuits to the asymptotic controllable regime. Finally, the dependence of the reachability bound on the initial target probability suggests an algorithmic application. The preparation of the state and the design of the mixer can be compared by how much probability they initially place in low- or high-cost subspaces, since increasing that weight directly weakens the depth obstruction.

\section*{Acknowledgements}
We thank Boris Tsvelikhovskiy for fruitful discussions.
BNB was supported by the U.S. Department of Energy, Advanced Scientific Computing Research, under contract number DE-SC0025384. DG is partially supported by Simons Collaboration Grant 855678. The calculations and the initial draft of this paper were produced with the assistance of Anthropic Claude Fable 5 and OpenAI ChatGPT 5.6. The authors assume full responsibility for the accuracy of the final version.


\end{document}